\documentclass[11pt]{mysodaclass}

\newtheoremstyle{slplain}
  {1mm}
  {3mm}
  {\it}
  {}
  {\bf}
  {}
  { }
  {}

\theoremstyle{slplain}

\usepackage{bei_defs}

\usepackage{jacobi_defs}

\title{Local, Smooth, and Consistent Jacobi Set Simplification}
\author{
Harsh Bhatia \\ {\small \textsl{hbhatia@sci.utah.edu}}  \\ {\small \textsl{SCI Institute, University of Utah}}\and 
Bei Wang \\ {\small\textsl{beiwang@sci.utah.edu}} \\ {\small \textsl{SCI Institute, University of Utah}}\and 
Gregory Norgard \\{\small\textsl{gregnorgard@gmail.com}} \\ {\small \textsl{Numerica Corporation}}\and
Valerio Pascucci  \\{\small\textsl{pascucci@sci.utah.edu}} \\ {\small \textsl{SCI Institute, University of Utah}}\and
Peer-Timo Bremer \\ {\small\textsl{bremer5@llnl.gov}}  \\ {\small \textsl{Lawrence Livermore National Laboratory}}
}
\date{}

\begin{document}

\begin{titlepage}
\maketitle 
\begin{abstract}
The relation between two Morse functions defined on a common domain
can be studied in terms of their Jacobi set. The Jacobi set contains
points in the domain where the gradients of the functions are
aligned. Both the Jacobi set itself as well as the segmentation of the
domain it induces have shown to be useful in various applications.
Unfortunately, in practice functions often contain noise and
discretization artifacts causing their Jacobi set to become
unmanageably large and complex. While there exist techniques to simplify
Jacobi sets, these are unsuitable for most applications as they lack
fine-grained control over the process and heavily restrict the type of
simplifications possible.

In this paper, we introduce a new framework that generalizes critical
point cancellations in scalar functions to Jacobi sets in two
dimensions.  We focus on simplifications that can be realized by
smooth approximations of the corresponding functions and show how this
implies simultaneously simplifying contiguous subsets of the Jacobi
set.  These extended cancellations form the atomic operations in our
framework, and we introduce an algorithm to successively cancel
subsets of the Jacobi set with minimal modifications according to some
user-defined metric.  We prove that the algorithm is correct and
terminates only once no more local, smooth and consistent simplifications are possible. 
We disprove a previous claim on the minimal Jacobi set for manifolds with arbitrary genus and show that for simply connected domains, our algorithm reduces a given Jacobi set to its simplest configuration.
\end{abstract}
\end{titlepage}
\pagestyle{plain}

\section{Introduction}
\label{sec:introduction}

In scientific modeling and simulation, one often defines multiple functions, 
e.g.\ temperature, pressure, species distributions etc.\ on a common domain. 
Understanding the relation between such functions is crucial in data exploration 
and analysis. The \emph{Jacobi set}~\cite{EdeHar2002} of two scalar functions 
provides an important tool for such analysis by describing points in the domain 
where the two gradients are aligned, and thus partitioning the domain into
regions based on relative gradient orientation. A variety of interesting physical 
phenomena such as the interplay between salinity and temperature of water in 
oceanography~\cite{BPS01} and the critical paths of gravitational potentials of 
celestial bodies~\cite{SM98} (similar to the Lagrange points in astrophysics) can 
be modeled using Jacobi sets. In data analysis and image processing, Jacobi sets 
have been used to compare multiple scalar functions~\cite{EdeHarNat2004}, as well 
as to express the paths of critical points overtime~\cite{BreBriDuc2007,EdeHar2002}, 
silhouettes of objects~\cite{Edelsbrunner09}, and ridges in image data~\cite{Norgard11cagdb}. 

However, the Jacobi sets can be extremely detailed to the point at which their 
complexity impedes or even prevents a meaningful analysis. Often, one is not 
interested in the fine-scale details, e.g.\ minor silhouette components due to 
surface roughness, but rather in more prevalent features such as significant 
protrusions. The Jacobi sets are also highly sensitive to noise which further leads to 
undesired artifacts. Finally, the most common algorithm to compute Jacobi
sets~\cite{EdeHar2002,Norgard11cagdb} is designed for piecewise linear functions 
defined on triangulations, and is well known to introduce a large number of 
discretization artifacts which could skew the analysis. The natural answer to these 
problems is the controlled simplification of a Jacobi set by ranking and ultimately 
removing portions of it in order of importance.

Some previous techniques exist that can be broadly classified into direct and indirect 
Jacobi set simplification. \emph{Indirect} 
simplification~\cite{BreBriDuc2007,LuoSafWan2009} simplifies the underlying functions 
in a hope to obtain a structurally and geometrically simpler Jacobi set. However this 
poses several problems. First, especially in the case of two non-trivial functions, changing
either of them can introduce a large number of complex changes in the Jacobi set. 
These changes are difficult to predict and track, and instead the Jacobi set is typically 
re-computed at each step, which quickly becomes costly. Second, the Jacobi set encodes 
the relation between two functions and therefore simplifying one function may not
actually simplify the Jacobi set. For example, two functions with complex gradient flows, 
which are similar in terms of relative orientation, define a small and simple Jacobi set. In 
this case, smoothing the gradient flow of either of the functions can introduce
significant additional complexity into the Jacobi set. Finally, creating an appropriate metric 
to rank potential simplification steps can be challenging as small changes relative to 
traditional function norms e.g.\ $L_2$ or $L_\infty$ may induce large changes in the Jacobi
set and vice versa.  

Alternatively, {\it direct} simplification aims to identify and remove ``unimportant'' portions 
of the Jacobi set and subsequently to determine the necessary changes in the corresponding 
functions. Such techniques are designed to reduce the complexity of a Jacobi set measured 
by a user-defined metric.  The first step~\cite{NNat2011} proposed in this direction views the 
Jacobi set as the zero level set of a complexity measure 
\cite{EdeHarNat2004} and removes components (i.e.\ loops) of the level set in order of their 
hyper-volume.  However, this strategy is limited to removing entire loops of the Jacobi set.  
In practice, much of the complexity of the Jacobi set is due to small undulations in the
level sets of the functions causing zig-zag patterns.  Such features are not addressed directly 
by a loop removal, which severely limits the usability of this approach.  Furthermore, as discussed 
in Section~\ref{sec:jssimp}, one can easily construct cases where loops should be combined rather 
than removed.  In contrast, our goal is to obtain a Jacobi set with fewer \emph{birth-death} (BD) points 
(where the level sets of the two functions and the Jacobi curve have a common normal direction) 
and fewer loops.

\newparagraph{Contributions} 
To overcome the current limitations in Jacobi set simplification, we introduce a new direct 
simplification framework for Jacobi sets of two Morse functions defined on a common smooth, 
compact, and orientable $2$-manifold without boundary. By extending the notion of critical
point cancellations in scalar fields to Jacobi sets, we identify all possible simplifications that are 
realizable by smooth approximations of the corresponding functions. Based on a user-defined 
metric, we then rank these operations and progressively simplify the Jacobi set until no further 
reduction is possible. Our framework provides a fine-grained control over a very general set of 
possible simplifications and allows, e.g.\ the combination of loops and the removal of
zig-zag patterns along side the traditional loop removal. In particular:

\begin{itemize}
\item We introduce the notion of \emph{local} pairings of points in the Jacobi set that can be 
	cancelled. These point-wise cancellations are then extended to contiguous sub-domain 
	bounded by segments of the Jacobi set, referred to as \emph{Jacobi regions}, which are 
	simplified simultaneously in a \emph{consistent} manner. To obtain \emph{smooth} realization 
	of the simplification, the modification of Jacobi regions are extended to collection of adjacent 
	regions, referred to as \emph{Jacobi sequences}. Each such sequence is a contiguous 
	sub-domain ranked by a user-defined metric and is simplified as one atomic operation;
\item We propose a simplification algorithm that constructs and successively cancels Jacobi 
	sequences. Our approach naturally cancels critical points of both functions, removes and 
	combines loops, straightens the Jacobi set by removing zig-zag patterns, and always reduces 
	the number of BD points;
\item We show that the algorithm is correct, which means that the simplified Jacobi set is valid, 
	and that it terminates only when no more local, smooth, and consistent simplifications are possible;
\item We disprove a previous claim on the minimal Jacobi set for manifolds with arbitrary genus and 
	show that for domains with even genus there always exist function pairs that create a single loop in 
	the Jacobi set; and
\item We show that for simply-connected domains, our algorithm reduces a given Jacobi set to its 
	minimal configuration; while for non-simply-connected domains, we discuss some fundamental 
	challenges in Jacobi set simplification.
\end{itemize}

\section{Background and Related Work}
\label{sec:fundamentals}

This section presents the relevant background on Morse 
theory~\cite{Matsumoto2002, Milnor1963} and Jacobi sets~\cite{EdeHar2002}, and 
discusses the existing Jacobi set simplification schemes. In the following, let 
$\Mspace$ be a smooth, compact, and orientable $2$-manifold without boundary.

\newparagraph{Morse functions} 
Given a smooth function $f: \Mspace \to \Rspace$, a point $x \in \Mspace$ is called a
\emph{critical point} if the gradient $\grad f$ of $f$ at $x$ equals zero, and the value
of $f$ at $x$ is called a \emph{critical value}. All other points are \emph{regular points} 
with their function values being \emph{regular values}. A critical point $x$ is 
\emph{non-degenerate} if the Hessian, i.e.\ the matrix of second partial derivatives at the 
point, is invertible. $f$ is a \emph{Morse function} if (a) all its critical points are 
non-degenerate and (b) all its critical values are distinct.

 \begin{figure}
\centering
	\includegraphics[width=0.49\textwidth]{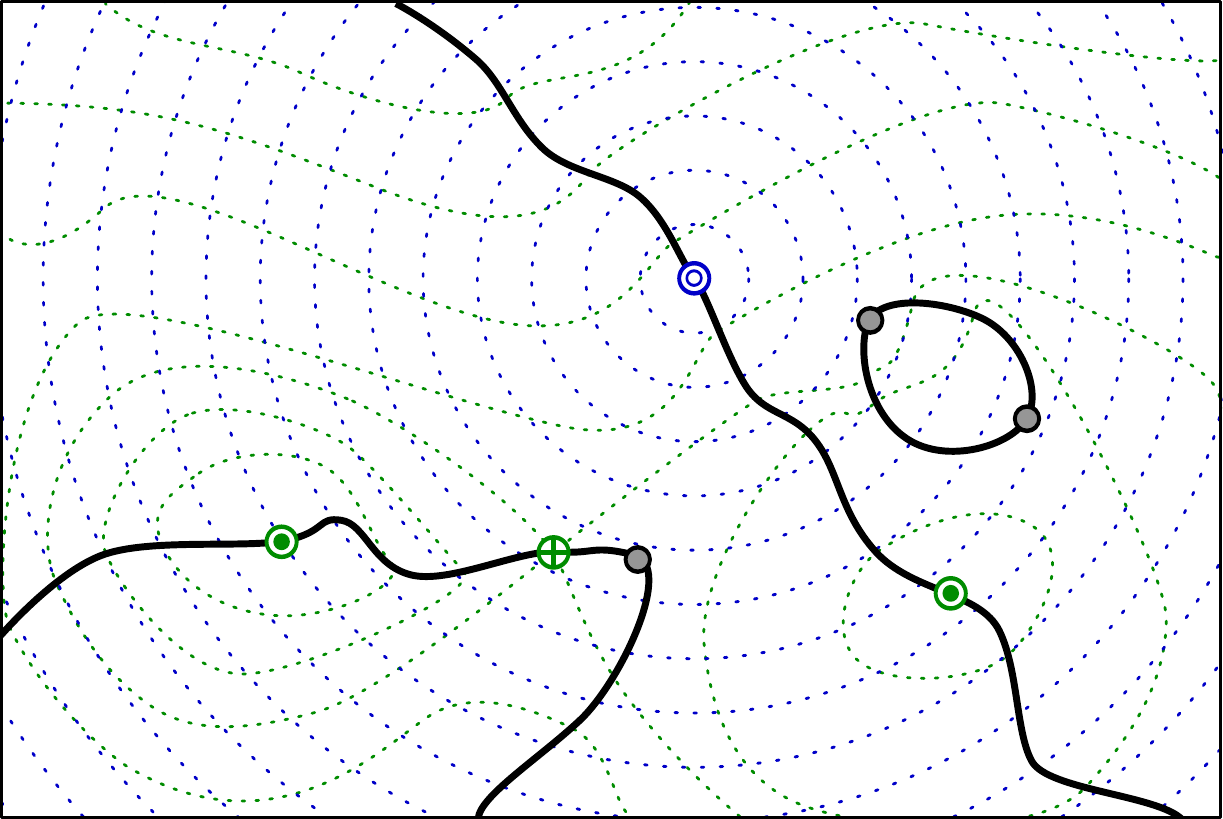} 
\caption{Jacobi set (solid black) of two functions with
    BD points shown in grey and critical points of the function in blue
    and green respectively. \label{fig:jacobi}}
\end{figure}

\newparagraph{Jacobi set} 
Given two generic Morse functions $f, g: \Mspace \to \Rspace$ such that the intersection 
of the sets of their critical points is a null set, their Jacobi set $\J = \J(f,g) = \J(g,f)$ is the 
closure of the set of points where their gradients are linearly dependent~\cite{EdeHar2002},
\begin{eqnarray}
\label{eq:jacobi:grad}
\J = \closure{\{x \in \Mspace \mid \grad f(x) + \lambda \grad g(x) = 0 \mbox{ or } \grad g(x) + \lambda \grad f(x) = 0\}}.
\end{eqnarray}

\noindent 
{The sign of $\lambda$ for each $x$ is also called as its \emph{alignment}, as 
it defines whether the two gradients are aligned or anti-aligned. }
By definition, the Jacobi set contains the critical points of both $f$ and $g$. 
Let $g^{-1}(t)$ represent the {\it level sets} of $g$ for $t \in \Rspace$, and 
$f_t: \lset g t \to \Rspace$ the restriction of $f$ on the level sets of $g$. 
Equivalently, Jacobi set can be defined as the closure of the set of critical points of $f_t$ 
for all regular values $t$ of $g$~\cite{EdeHar2002}.
\begin{eqnarray}
\label{eq:jacobi:rcp}
\J = \closure{\{x \in \Mspace \mid x \mbox{ is critical point of } f_t\}}.
\end{eqnarray}

\noindent The critical points of $f_t$ are also referred to as the \emph{restricted 
critical points} of $f$ (with respect to $g$). The restricted function $f_t$ is a Morse 
function almost everywhere\footnote{The set of points where the function $f_t$ is not 
Morse is a finite set of measure zero.}. Three types of degeneracies exist where $f_t$ 
is not Morse for some $t \in \Rspace$: (a) $t$ is a critical value of $g$, then the level 
set $\lset g t$ contains a singularity and thus is not a $1$-manifold; (b) Two or more 
critical points in $f_t$ share the same function value; (c) $f_t$ contains an inflection 
point (a degenerate critical point). These degeneracies play an important role in our 
discussion on Jacobi set simplification. For example, each restricted critical point along 
$\J$ is an extrema of $f_t$ for some $t \in \Rspace$. As $t$ varies, maxima and minima
of $f_t$ can approach each other and ultimately merge at an inflection point called a 
\emph{birth-death (BD) point}. Alternatively, traveling along $\J$, restricted critical points of
$f_t$ switch their criticality (from maximum to minimum or vice versa) at BD points. 
Furthermore, the restricted functions $f_t$ switch criticality at critical points of $g$ (but 
not at critical points of $f$). {Similarly, the alignment of restricted critical points 
switches at critical points of both $f$ and $g$. }
Figure~\ref{fig:jacobi} illustrates these concepts.

\newparagraph{Comparison measure} Several other descriptions of Jacobi sets 
exist~\cite{EdeHar2002,EdeHarMas2008,EdeHarNat2004,NNat2011}.  
One such description is in terms of a gradient-based metric to compare two
functions, called the \emph{comparison measure} $\kappa$~\cite{EdeHarNat2004}. 
It plays a significant role in assigning an importance to subsets of a Jacobi set in 
terms of the underlying functions $f$ and $g$ by measuring the relative orientation 
of their gradients.  For a domain $\Omega$,
\begin{eqnarray*}
\kappa  = \kappa(\Omega) = \frac{1}{\area{(\Omega)}} \int_{x \in \Omega} \kappa_x \di x 
= \frac{1}{\area{(\Omega)}} \int_{x \in \Omega} ||\grad{f(x)} \times \grad{g(x)}|| \di x,
\end{eqnarray*}
where $\di x$ is the area element at $x$, and $\area{(\Omega)} = \int_{x \in \Omega}\di x$. 
Here $\kappa_x = ||\grad{f(x)} \times \grad{g(x)}||$ represents the limit of 
$\kappa=||\grad f \times \grad g||$ to a single point, and the Jacobi set is its $0$-level 
set~\cite{EdeHarNat2004,LuoSafWan2009,NNat2011}.

\newparagraph{Level set neighbors}
For a point $v \in \J$, we can define its \emph{level set neighbors} (with respect 
to $g$), $\n_g(v)$, as neighbors of $v \in \J$ along $g^{-1}(g(v))$.  Two points 
$u,v \in \J$ are level set neighbors if $u \in \n_g(v)$ which implies $v \in \n_g(u)$. 
Generically, $|\n_g(v)|\le2$, however, for an extremum of $g$, $|\n_g(v)| = 0$, 
and for a saddle of $g$, $|\n_g(v)| \le 4$.  The level set neighbors can symmetrically 
be defined with respect to $f$. Such a definition can be extended to smooth curves in 
$\J$. Two smooth parametrized curves $\alpha, \beta: (a,b) \to \Mspace$ in $\Jspace$
are level set neighbors if $\alpha(t)$ and $\beta(t)$ are level set neighbors in 
$\lset g t$ for all $t \in (a,b)$. For simplicity in notations, for such level set neighbors, 
we choose $a$ and $b$ to be function values of $g$, i.e.\ $g(\alpha(a)) = a$ and
$g(\alpha(b)) = b$. We further define their \emph{bounded region}, denoted by 
$R_{(a,b)}(\alpha, \beta)$, as the open subset of $\Mspace$ bounded by curves $\alpha$, 
$\beta$, and level sets of $g$ that pass through their end points, i.e.\ $g^{-1}(a)$ and
$g^{-1}(b)$.\footnote{In the case where $\alpha, \beta: (a, b) \to \J$ are subsets of some 
larger parametrized curves $\alpha', \beta': (a',b') \to \J$, that is, $\alpha, \beta$ are the 
restriction of $\alpha', \beta'$ to $(a,b) \subseteq (a',b')$, i.e. $\alpha = \alpha' |_{(a,b)}$ 
and $\beta = \beta'|_{(a,b)}$, we denote their bounded region as $R_{(a,b)}(\alpha', \beta')$.}


\subsubsection*{Related Work}
As discussed above, the Jacobi set may contain a number of components that 
represent noise, degeneracies, or insignificant features in the data. As a result, 
Jacobi set simplification is both necessary and desirable. Bremer et 
al.~\cite{BreBriDuc2007} use the Jacobi set to track the critical points of a 
time-varying function $f: \Mspace \times \Rspace \to \Rspace$, where time is 
represented as $g: \Mspace \times \Rspace \to \Rspace$ and $g(x,t) = t$. 
The Jacobi set $\J = \J(f,g)$ is therefore the trajectory of the critical points of 
$f_t$ as time varies. To simplify the Jacobi set, they use the Morse-Smale
complex of $f_t$ at discrete time-steps to pair critical points, cancel pairs below 
a persistence threshold, and remove small components of the Jacobi set that lie 
entirely within successive time-steps. This method, however, is difficult to extend 
to a general setting: First, only one function, $f$, is simplified and the other is 
assumed to be trivial; Second, only a small, discrete number of $f_t$ are simplified 
and all intermediate changes are ignored.

Luo et al.~\cite{LuoSafWan2009} propose an algorithm to compute the Jacobi set 
of a point cloud. The Jacobi set is considered as the $0$-level set of $\kappa_x$, 
which is computed by approximating the gradients $\grad f$ and $\grad g$. 
Reducing the number of eigenvectors used in the gradient approximation, therefore, 
corresponds to a simpler Jacobi set after re-computation. This is the foremost 
example of an indirect simplification in which $f$ and $g$ are smoothed which leads 
to some (unpredictable) changes in $\J$. Instead, as discussed below, this paper 
aims at identifying and removing an unimportant portion of $\J$ by determining how 
$f$ and $g$ can be modified accordingly.

N and Natarajan~\cite{NNat2011} consider the simplification of the Jacobi set as the 
reduction in the number of components in $\J$ with minimal change to the relationship 
between the two functions, quantified by $\kappa_x$. The authors construct the Reeb
graph~\cite{ShinagawaKunii1991} of $\kappa_x$, and associate a percentage of $\kappa$ 
as offset cost with each critical point and $0$-level set point in the Reeb graph. A greedy 
strategy is then applied to modify a component in the Jacobi set with the least offset cost 
until a threshold is reached. However, this technique is restricted to removing entire loops 
of $\J$, which significantly restricts its flexibility. For example, one can easily construct
examples where $\J$ is highly complex yet contains only a single loop. 


\section{Jacobi Set Simplification -- An Overview}
\label{sec:jssimp}

As discussed in Section~\ref{sec:introduction}, this paper introduces a direct 
simplification of the Jacobi set, i.e.\ it removes a given set of points from $\J$, 
by understanding the required changes in $f$ and/or $g$. The goal is to obtain 
a Jacobi set with fewer BD points and/or fewer loops, by making the gradients of 
the underlying functions more similar. In the following, we describe simplification 
of $\J$ which modifies $f$ with respect to the level sets of $g$, but all concepts 
apply symmetrically to modifications of $g$ with respect to $f$. In practice, we 
consider simplifications that modify either $f$ or $g$, and typically interleave 
operations acting on one or the other.

Since the Jacobi set $\J(f,g)$ is defined as the closure of the (restricted) critical 
points of $f_t$ for regular values $t \in \Rspace$ where $f_t$ is a $1$D 
function, it is natural to simplify $\J$ by canceling restricted critical points in $f_t$. 
In the topological simplification of a scalar function, typically the features
of interest are critical points. Based on the Morse Cancellation Theorem
(Theorem 5.4 in \cite{Milnor1965} as the First Cancellation Theorem, or
\cite{Morse1965}), critical points must be removed in pairs through atomic
{cancellation} operations. 
Therefore, we remove 
pairs of restricted critical points 
to construct continuous 
simplified function $\overline {f_t}$, such that no other critical points of $f_t$ are 
affected. To obtain a smooth approximation $f_t^*$ of the simplified function 
$\overline{f_t}$, the modification can be extended to allow an $\epsilon$-slope 
for the modified function. 
The region of influence of this cancellation is the region 
where $f_t \neq f_t^*$, and is highlighted in Figure~\ref{fig:overview_ft}. In order 
to perform these cancellations, we must first define a scheme for pairing restricted 
critical points. Section~\ref{sec:js:rcp} discusses in detail the choice of our pairing 
scheme, and the procedure of carrying out such cancellations. 

Although cancellation of restricted critical points produces smooth simplified restricted 
functions $f_t^*$ as shown in Figure~\ref{fig:overview_ft}, performing a single such cancellation 
creates a discontinuity across the level set \lset{g}{t}. 
In order to obtain smoothness across level sets, we must extend these 
cancellations by canceling more than one contiguous pairs of restricted critical points, called 
the \emph{Jacobi regions}, at the same time. For example, consider two Jacobi regions $V_1$ 
and $V_2$ existing between the level sets \lset{g}{a}, \lset{g}{c}, and \lset{g}{b} as shown in 
Figure~\ref{fig:valid} (left). These Jacobi regions represent contiguous pairs of restricted critical 
points shown as red and blue lines respectively. A smooth simplification $f^*$ that cancels all 
critical points in these regions can be obtained by modifying $f$ in the corresponding shaded 
region. The construction, properties, and cancellation of the {Jacobi regions} are discussed in 
detail in Section~\ref{sec:js:regions}.

\begin{figure*}[b]
\centering
\subfigure[]{
 	\def\svgwidth{0.6\linewidth} 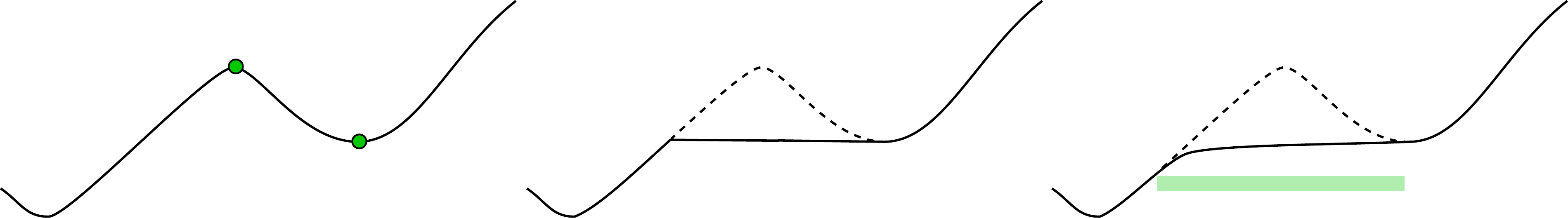
\label{fig:overview_ft}}
\hspace{0.6cm}
\subfigure[]{
	\def\svgwidth{0.3\linewidth} 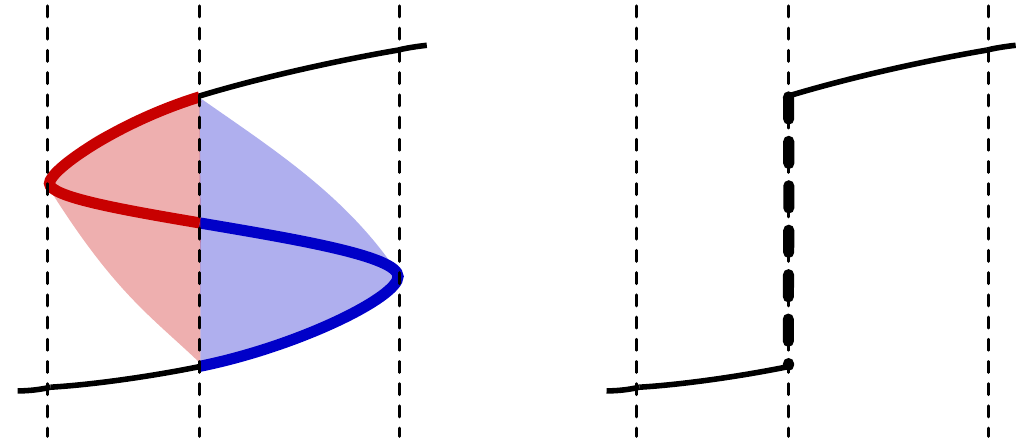\label{fig:valid}}
\caption{
(a) Cancellation of a pair of critical points in the 1-dimensional $f_t$ (left) gives a continuous simplified 
simplified $\overline{f_t}$ (middle). A smooth simplification can be obtained as $f_t^*$. The cancellation 
modifies $f_t$ in the shaded region.  
(b) (Left) Valid simplification of $V_1$ and $V_2$ must construct smooth $f^*$ by 
modifications in the corresponding shaded regions only. $a$, $b$, and $c$ represent 
level sets of $g$. (Right) The simplified Jacobi set.
\label{fig:overview}}
\end{figure*}

The cancellation of Jacobi regions, however, creates smooth functions only at the interior of 
the regions, and the discontinuities are pushed to their boundaries. In order to create globally 
smooth simplified functions $f^*$, we must further cancel a sequence of adjacent regions at 
the same time, e.g.\ canceling $V_1$ and $V_2$ at the same time as shown in 
Figure~\ref{fig:valid} (right). We show that any discontinuities can be avoided by local 
modifications if these \emph{Jacobi sequences} start and end with BD points, and discuss 
their construction and cancellation in Section~\ref{sec:js:seq}. 

This way, the entire Section~\ref{sec:js:main} focuses on a simplification scheme that 
extends the concept of critical point cancellation in scalar functions to Jacobi sets. The 
defining characteristic of a \emph{valid simplification} is the removal of pairs of restricted 
critical points in $\J$ in a \emph{local}, \emph{smooth}, and \emph{consistent} manner.

\begin{definition}[Valid Simplification] \upshape{
  {Let $V$ be a set of level set neighbors in $f_{[a,b]} = \{f_t \mid t \in [a,b]\}$ 
  for some interval $[a,b] \subseteq \Rspace$. Removing $V$ from $\J$ is considered a
  \emph{valid simplification} if it is }
  \begin{enumerate}\dense
  \item[1.] \emph{local}: There exists a continuous
   $\overline{f}_{[a,b]} = \{\overline{f}_t \mid t \in [a,b]\}$ containing all critical points of 
   $f_{[a,b]}$ except for $V$;
  \item[2.] \emph{smooth}: There exists a smooth $f^*:\Mspace \rightarrow \Rspace$ 
  such that $||f^*_{[a,b]} - \overline{f}_{[a,b]}||_\infty < \varepsilon$ for any 
  $\varepsilon > 0$; and
  \item[3.] \emph{consistent}: $\J(f^*,g) = \J(f,g)$ for all $x$ with 
  $g(x) \in (-\infty,a) \cup (b,\infty)$, and $f^*(x) = f(x)$ for all $x$ with 
  $g(x) \in (-\infty,a-\epsilon] \cup [b+\epsilon,\infty)$ for any $\epsilon > 0$.
  \end{enumerate}}
\end{definition}

\noindent 
We point out that the simplified function $f^*(x)$ is defined with respect to a given 
$\epsilon$. 
Referring to Figure~\ref{fig:overview_ft} it is important to note that the locality conditions 
implies that the modification in $f_t$ must not impact any restricted critical points other 
than $u$ and $v$. In Figure~\ref{fig:valid}, this means that for any level set of $g$ 
(vertical line), the red and blue shaded regions must not touch any portion of Jacobi set 
other than the ones shown in red and blue respectively. Notice that while locality is associated 
with continuous function $\overline{f_t}$, the second condition requires $f^*$ to be smooth 
along and across level sets. In order to create such a smooth $f^*$, the locality condition 
must be relaxed within a small neighborhood around the cancelled Jacobi region. Furthermore, 
the consistency condition requires that no portions of Jacobi set outside $[a,b]$ are modified. 
In a way, the consistency condition implies locality across level sets. 
%
While the locality condition is obtained by defining a special pairing 
function, a smooth and consistent simplification can be performed when the Jacobi sequence 
begins and ends with BD points (as in Figure~\ref{fig:valid}).
Also, notice that 
$\J(f^*,g) \not\subset \J(f,g)$, since new points (dashed line) may be added to the Jacobi set 
to connect the existing curves.

Unfortunately, as detailed in Section~\ref{sec:js:seq}, the saddles of $g$ present unresolvable 
discontinuities in the pairings, and therefore may obstruct the construction of Jacobi sequences. 
Consequently, the simplification scheme discussed above may not be able to progress. 
In order to handle such cases, 
we use a conventional critical point cancellation technique in $2$D to cancel a saddle 
of $g$ with its maxima/minima. As shown in Section~\ref{sec:saddle}, our approach does not 
change the Jacobi set structurally, but only simplifies (reduces the number of) its alignment switches. 

Using the simplification techniques discussed in Sections~\ref{sec:js:main} and~\ref{sec:saddle}, 
Section~\ref{sec:summary} presents a combined procedure for simplifying Jacobi sets which can 
be guided by an arbitrary metric. We provide correctness proofs for the procedure, and show 
that for simply-connected domains, this procedure obtains the simplest possible configuration of 
Jacobi sets. On the other hand, for non-simply-connected domains, we discuss current challenges 
and list them as future work. 


\section{Cancellation of restricted critical points in $f$}
\label{sec:js:main}

This section details the procedure of canceling restricted critical points in $\J$ 
to obtain simplified functions. Starting with the simplification of $1$D restricted 
functions, we discuss the cancellation of entire segments of $\J$ by canceling 
Jacobi sequences.

\subsection{Pairing and cancellation of restricted critical points}
\label{sec:js:rcp}

Since restricted critical points of $f_t$ must be cancelled in pairs, we need a 
mechanism to define such pairings. The 
\emph{topological persistence pairing}~\cite{EdeCohZom2002,ZomorodianCarlsson2005}
seems to be an obvious choice, where critical points are paired and removed 
in order of \emph{persistence}. However since persistence pairing is assigned 
globally, restricted critical points, which are not level set neighbors may be 
paired. These pairs cannot be cancelled without violating the locality condition, 
which prevents most simplifications. Therefore, we instead use a localized 
variant of persistence pairing that guarantees that each point on the Jacobi set 
is paired with one of its level set neighbors as described below.

Given a non-degenerate restricted critical point $v \in f_t$ and its two level set 
neighbors $u, w \in n_g(v)$, the goal is to understand how $f_t$ can be modified 
in a local neighborhood surrounding $v$, in order to cancel $v$ with either $u$ 
or $w$. Consider, e.g.\ $v_3$ shown in Figure~\ref{fig:ft}. One can lower 
$v_3$ to the level of $v_4$ canceling $(v_3,v_4)$, but cannot lower it to the level 
of $v_2$ as this would impact $v_4$, and thus become a non-local simplification. 
In general, each restricted critical point can be cancelled with only one of its level 
set neighbors in this fashion, and we call such a neighbor its {\it partner}. Formally, 
this relation between a restricted critical point and its partner can be described 
through a \emph{local pairing function}, $\p: \J \to \J$, such that for every $v \in \J$, 
its \emph{partner} $\p(v)$ is defined as  
(a) $v$, if $v$ is a degenerate critical point of $f_t$ or a critical point of $g$; or
(b) an arbitrary element in the set $\{u \mid \argmin_{u \in \n_g(v)} \|f_t(u)-f_t(v)\| \}$ 
otherwise. Intuitively, every non-degenerate restricted critical point $v$ is paired with 
one of its level set neighbors $u$ with minimal difference in function value. Then 
$(v, u)$ is referred to as a \emph{(local) pair}.  Notice that, $\p(v) = u$ does not imply 
$\p(u) = v$. Traveling along a Jacobi curve, the discontinuities of $\p(v)$ reflect a 
change in partner for $v$. Since BD points and extrema of $g$ are paired to themselves, 
$\p$ is continuous at such points. Figure~\ref{fig:ft} indicates the pairings between 
restricted critical points as directed arrows pointing from $v$ to its partner $\p(v)$.

\begin{figure*}[!t]
\centering
\subfigure[]{ 	\def\svgwidth{0.25\linewidth} 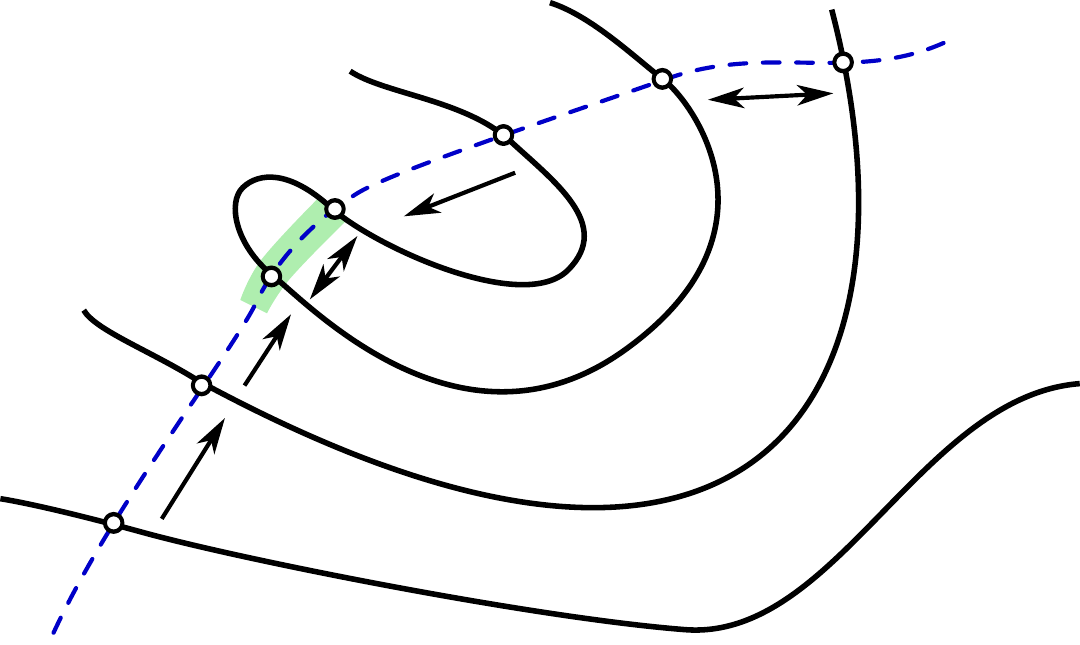 \label{fig:jacpairings}}
\subfigure[]{ 	\def\svgwidth{0.35\linewidth} 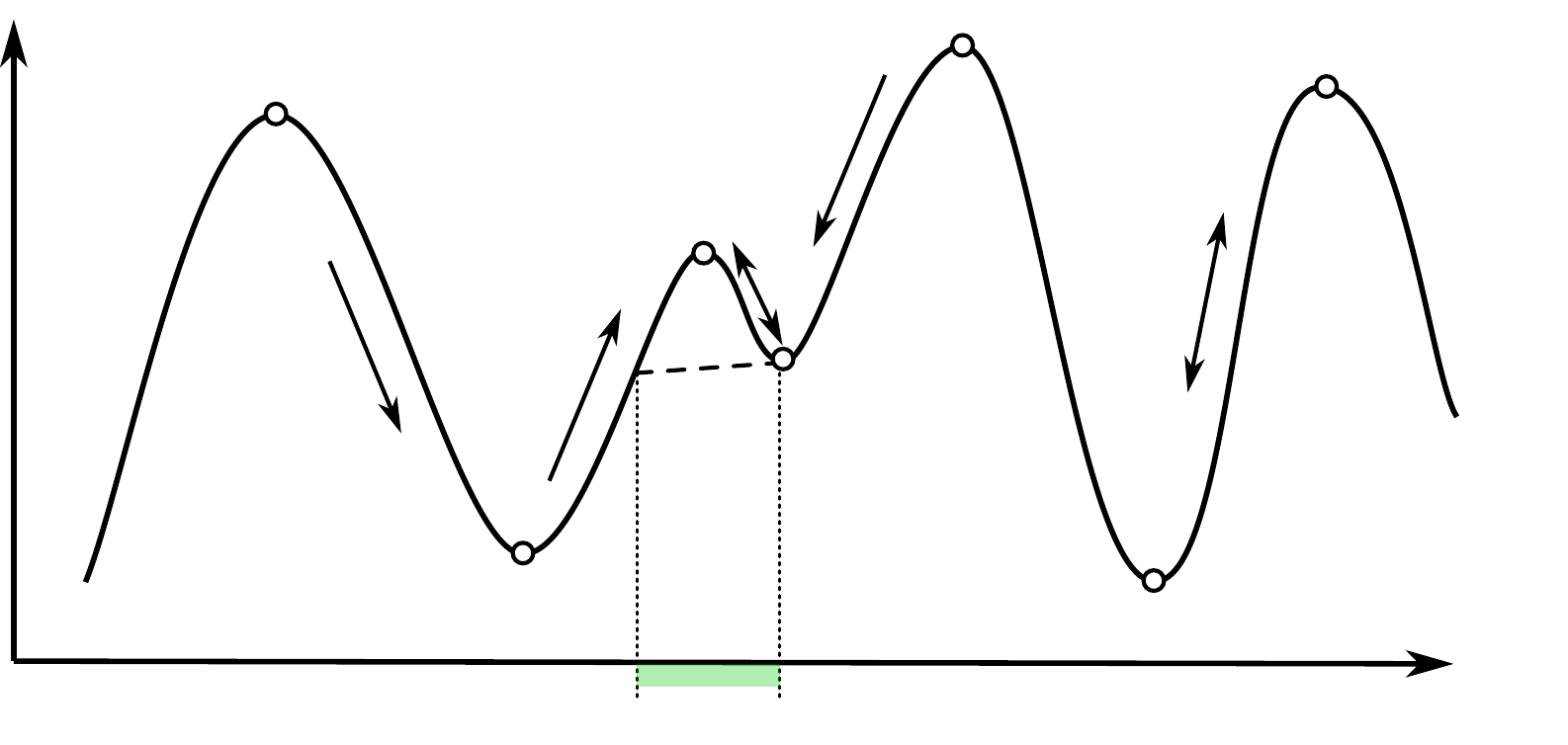 \label{fig:ftpairings}}
\caption{Illustration of restricted critical points and local pairings. 
(a) A Jacobi set $\J$ (black solid lines) intersects a level set $g^{-1}(t)$ (blue dashed line), and 
(b) the corresponding restricted function $f_t$ is shown. Local pairings among the restricted 
critical points in $f_t$ are indicated by arrows. The pair $(v_3,v_4)$ can be cancelled by 
lowering the maximum $v_3$ to match the value of $v_4$ (black dashed line in (b)). For the 
cancellation, its region of influence along the level set is shown in green (in both (a) and
(b)). It is the subset of the level set where the function value is modified, however, is 
shown as a thick region only for illustration. 
\label{fig:ft}}
\end{figure*}

Using the assigned local pairings between critical points of $f_t$, a simplification that removes 
critical points in pairs through atomic {cancellation} operations is given by the Morse Cancellation Theorem 
(Theorem 5.4 in~\cite{Milnor1965} as the First Cancellation Theorem, or~\cite{Morse1965}). 
We can perform such a cancellation by moving a critical point to the level of its partner to obtain a 
continuous simplified function $\overline{f_t}$, as shown in Figures~\ref{fig:overview_ft} 
and~\ref{fig:ftpairings}. To obtain a smooth $f_t^*$, an $\epsilon$-slope can be introduced 
in $\overline{f_t}$ while still maintaining locality. Consequently, for a pair $(v,u)$, a cancellation 
where $v$ is moved to the level of $u$ always guarantees locality. Notice in Figure~\ref{fig:ftpairings} 
that the pair $(v_5, v_6)$ could also be cancelled locally by bringing both points to a function value 
between $v_4$ and $v_7$. In general, one could potentially bring both points to a common 
intermediate value for a local cancellation. However, such cancellations may not admit valid 
simplification steps for reasons explained in Section~\ref{sec:js:regions}, and therefore are not 
considered. From now on, a cancellation induced by a pair $(v, u)$ always implies a procedure 
that moves $v$ to the level of $u$.
\subsection{Construction and cancellation of Jacobi regions}
\label{sec:js:regions}

The cancellation of a pair of restricted critical points creates a smooth restricted function $f_t^*$. 
However, the function is still discontinuous across the level set \lset{g}{t}, since the 
neighboring restricted functions are unchanged. Hence, canceling a single pair of 
restricted critical points in isolation introduces unwanted discontinuities, and 
therefore violates the smoothness condition of a valid simplification. Instead, one can 
extend these cancellations to adjacent restricted functions, which, however, violates 
the consistency condition of a valid simplification. For example, consider the scenario 
shown in Figure~\ref{fig:discontinuity2}. Canceling $(u,v) \in f_{t_0}$ creates a 
discontinuous simplified function. This modification can be extended to an adjacent 
region $f_{[t_0-\epsilon,t_0+\epsilon]}$ allowing the creation of a smooth function 
$f^*$ at $t_0$ which cancels $(u, v)$. However, since $\J$ is now modified beyond the 
level set $g^{-1}(t_0)$, it is no more a consistent simplification. 

\begin{figure*}[t]
\centering
\subfigure[]{ 	\def\svgwidth{0.265\linewidth} 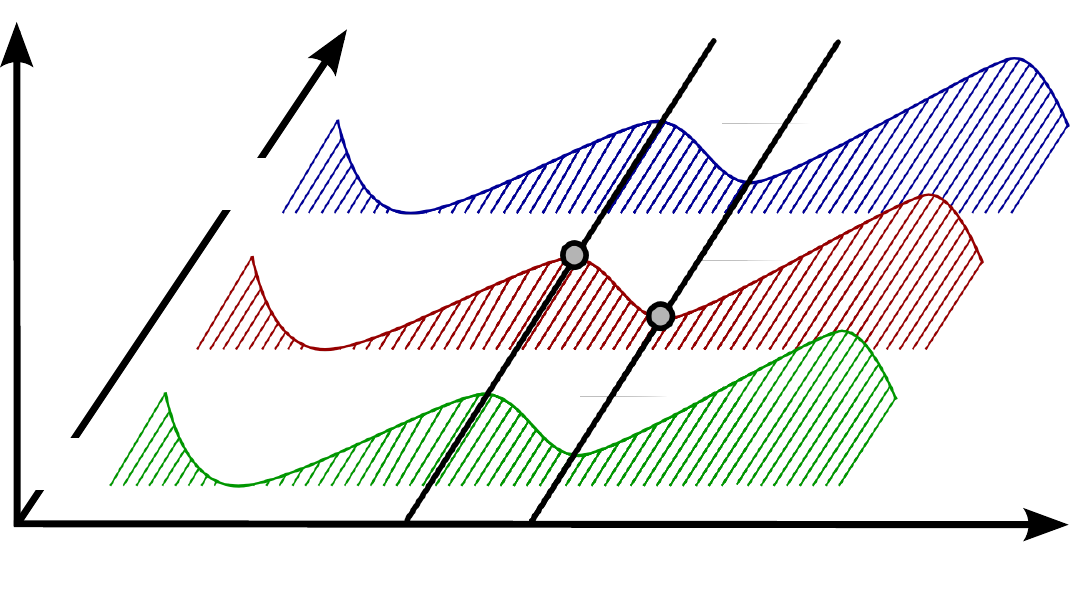 \label{fig:discont2_orig}}
\subfigure[]{ 	\def\svgwidth{0.265\linewidth} 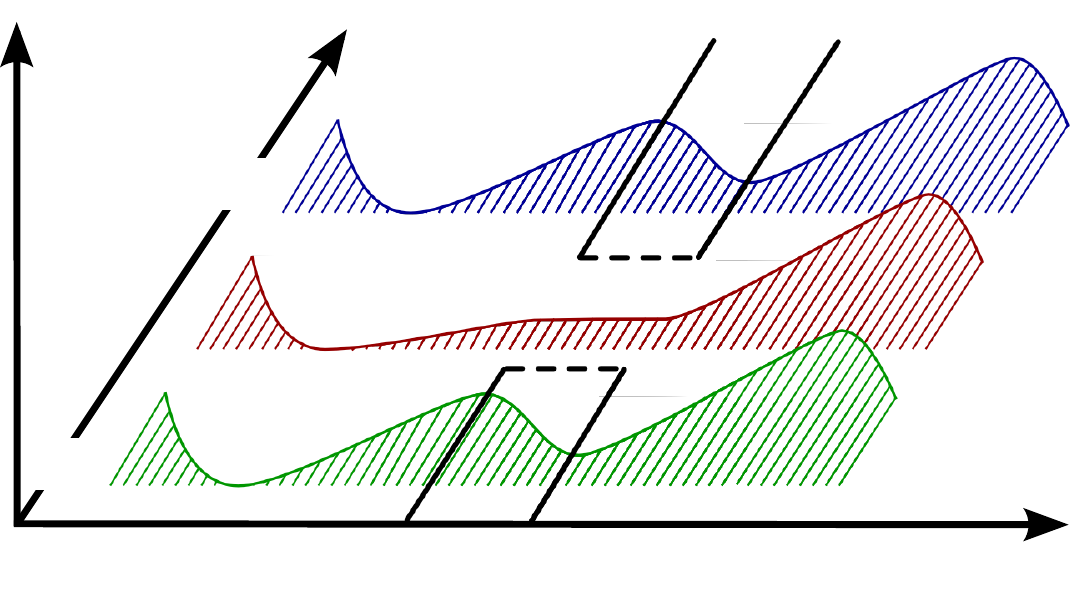 \label{fig:discont2_canc}}
\subfigure[]{ 	\def\svgwidth{0.265\linewidth} 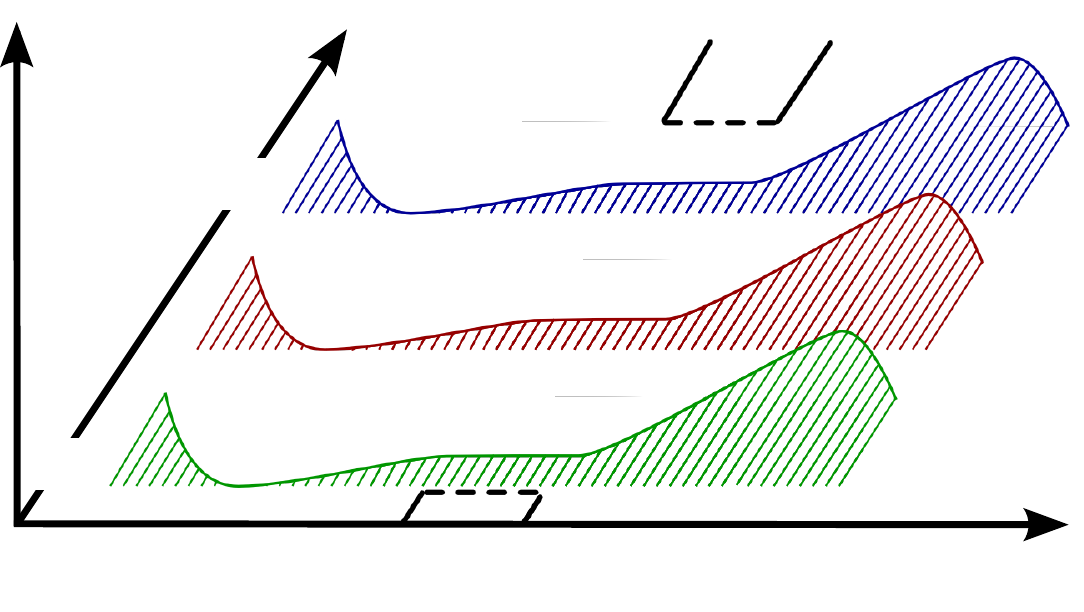 \label{fig:discont2_canc2}}
\caption{Cancellation of a pair of restricted critical points $(u,v) \in f_t$. 
(a)~The original $f_t$'s and  Jacobi set (in black). 
(b)~Canceling $(u,v)$ in $f_t$ in isolation creates a discontinuity across $t = t_0$, and hence is invalid. 
(c)~Extending the cancellation to $f_{[t_0-\epsilon,t_0+\epsilon]}$ creates a smooth $f^*$, but the cancellation is inconsistent since $\J$ outside $[t_0,t_0]$ is modified.
\label{fig:discontinuity2}}
\end{figure*}

Therefore, one must cancel connected sets of neighboring restricted critical points 
that are paired consistently. To understand their construction, we define 
\emph{switch points} as the set of points in $\J$ where $\p$ is not continuous, and 
\emph{boundary points} which are either switch points, BD points, or critical points of 
$g$. Then, the Jacobi set $\J$ can be decomposed into a set of non-overlapping
\emph{Jacobi segments}, which are maximal open subsets of $\J$ separated by 
boundary points. By definition, restricted critical points within the interior of Jacobi 
segments are consistently paired since $\p$ is continuous, and thus $\p$ induces 
a pairing between segments. Finally, we define {\it image points} as the 
level set neighbors of boundary points. Together, the boundary points and the image 
points decompose the Jacobi set into pieces $\alpha_i$ that have mutually consistent 
pairing, meaning that $\p$ is continuous both on $\alpha_i$ and its partner 
$\beta_i = \p(\alpha_i)$.  Given two such maximal subsets of Jacobi segments which 
are level set neighbors parametrized as $\alpha, \beta: (a,b) \to \J$, we call their
bounded region $R_{(a,b)}(\alpha,\beta)$ a {\it Jacobi region}. Similar to the point-wise 
cancellation, the entire segment $\alpha$ can be moved to the level of $\beta$ to cancel 
both the segments. Figure~\ref{fig:regions} shows boundary and image points, Jacobi segments, 
and Jacobi regions as pairings between them for a typical Jacobi set configuration.

\begin{figure}[t]
\centering
 	\def\svgwidth{0.3\linewidth} 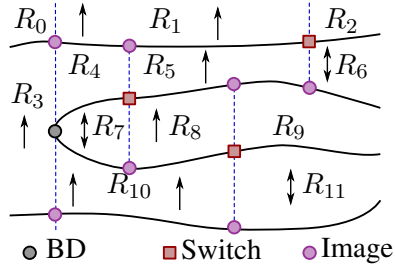
\caption{Illustration of Jacobi regions as pairings between Jacobi segments.\label{fig:regions}}
\end{figure}

There exist various classes of Jacobi regions with different implications on the simplification 
process. A Jacobi region is called {\it regular} if its closure does not contain BD points or critical
points of $g$. Regular regions have four ``corners'' made up of two switch and
two image points, e.g.\ $R_5, R_8$ in Figure~\ref{fig:regions}. With slight abuse of notation, 
we denote a corner as $\alpha(a) = \lim_{t \to a} \alpha(t)$. We further identify special but not 
mutually exclusive types of regions shown in Figure \ref{fig:cases}: 
(a) {\it BD internal regions} where $\alpha$ and $\beta$ share at least one BD point, i.e.\ $\alpha(a) = \beta(a)$ and/or $\alpha(b) = \beta(b)$;
(b) {\it BD side region} where $\alpha$ and/or $\beta$ are bounded by a BD point but $\alpha(x) \ne \beta(x)$, for all $x \in [a,b]$;  
(c) {\it BD external region} where the boundary of the region contains a BD point but neither $\alpha$ nor $\beta$ does; 
(d) {\it Saddle region} where the boundary of the region contains a saddle of $g$ but neither $\alpha$ nor $\beta$ does; and 
(e) {\it Extremal region} containing an extremum of $g$. 

\begin{figure*}[t]
\centering
\hspace{-4em}
\subfigure[]{  \def\svgwidth{0.22\linewidth}	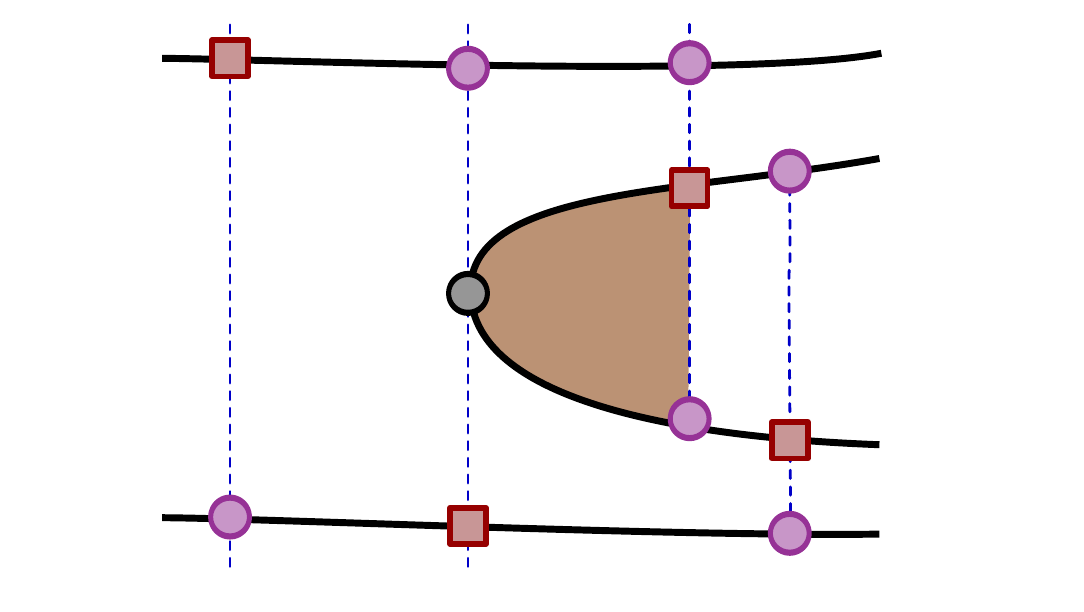 \label{fig:case_internal}}
\hspace{-2em}
\subfigure[]{  \def\svgwidth{0.22\linewidth}	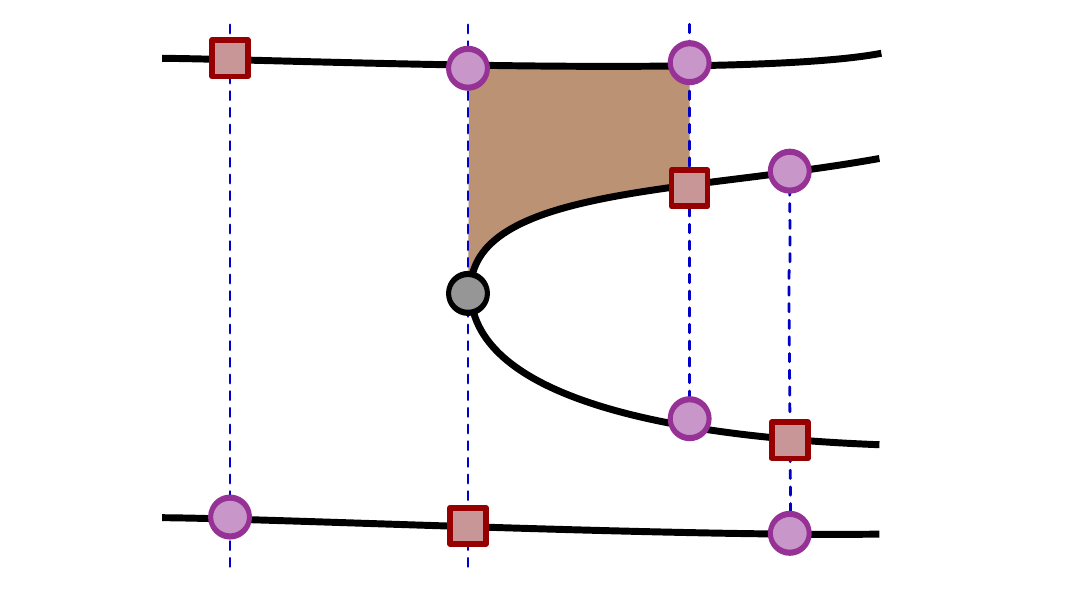 \label{fig:case_side}}
\hspace{-2em}
\subfigure[]{ 	\def\svgwidth{0.22\linewidth} 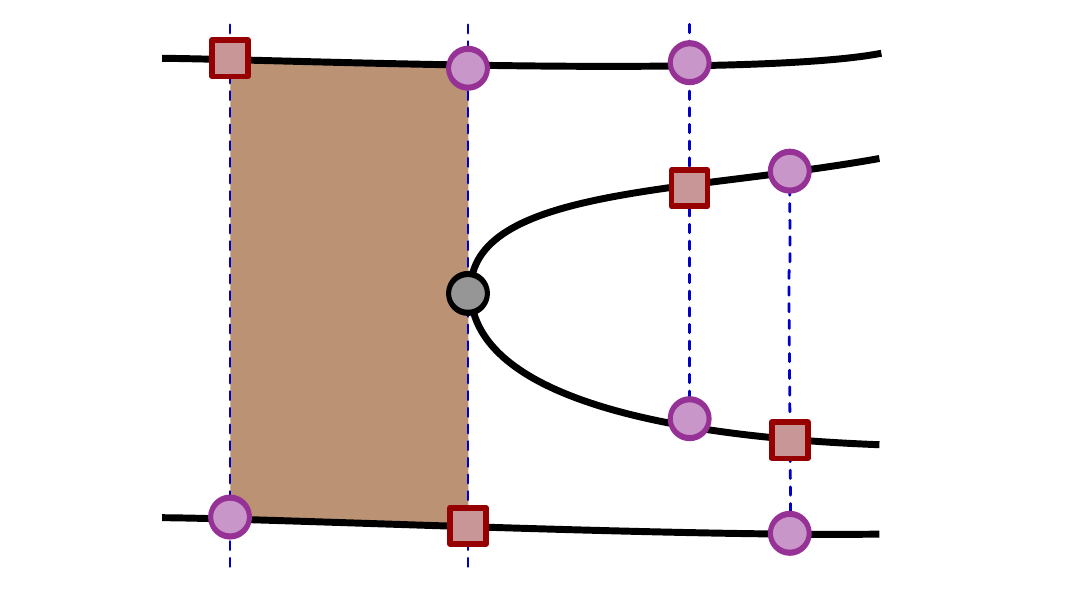 \label{fig:case_external}}
\hspace{-2em}
\subfigure[]{  \def\svgwidth{0.22\linewidth}	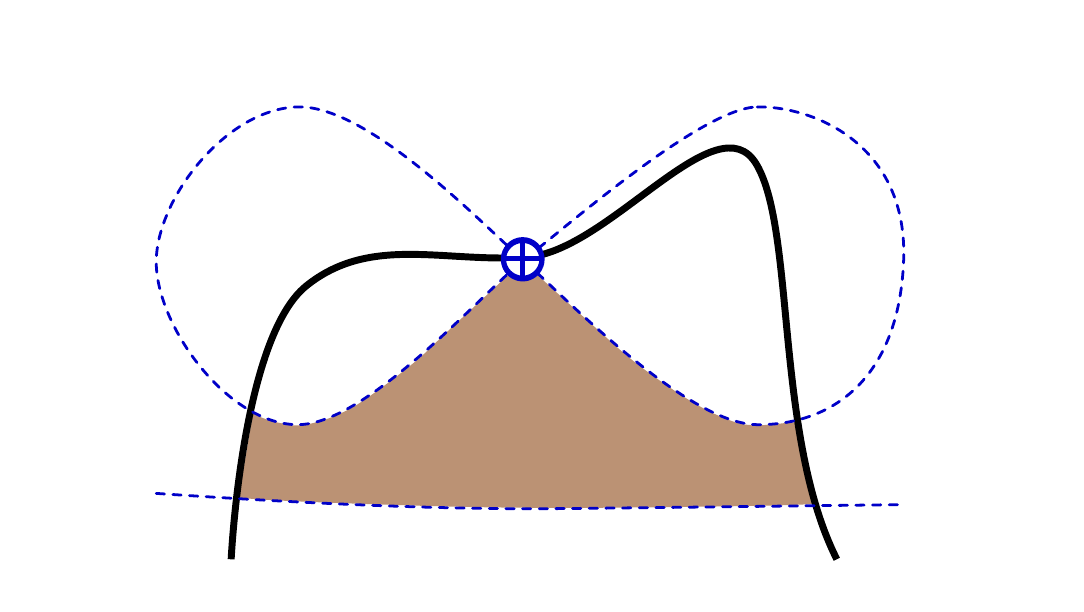 \label{fig:case_saddle}}
\hspace{-2em}
\subfigure[]{  \def\svgwidth{0.22\linewidth}	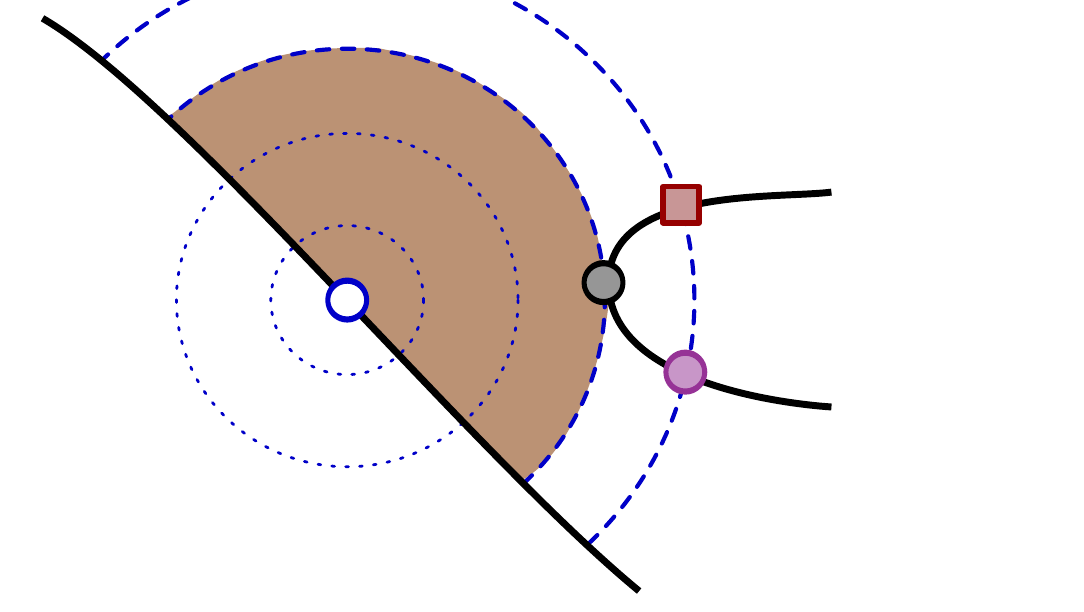 \label{fig:case_extremal}}
\hspace{-4em}
\caption{Special Jacobi regions: (a) BD internal; (b) BD side;  (c) BD external; (d)~Saddle; and (e) Extremal
  region.\label{fig:cases}}
\end{figure*}

\begin{figure*}[b]
\centering
\subfigure[]{	\def\svgwidth{0.27\linewidth} 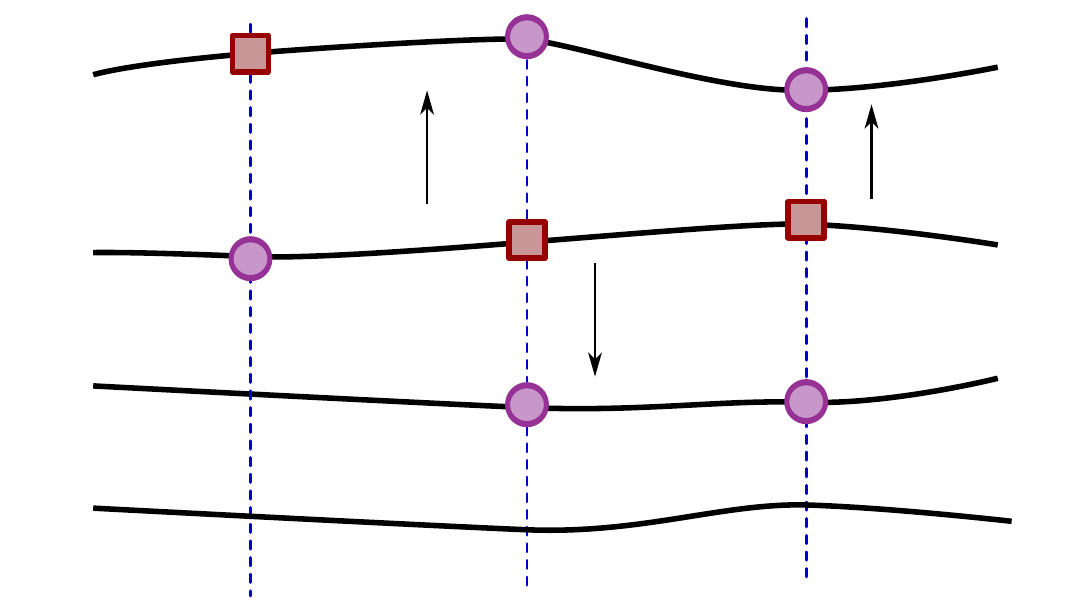 \label{fig:region_regular_orig}}
\subfigure[]{ 	\def\svgwidth{0.27\linewidth} 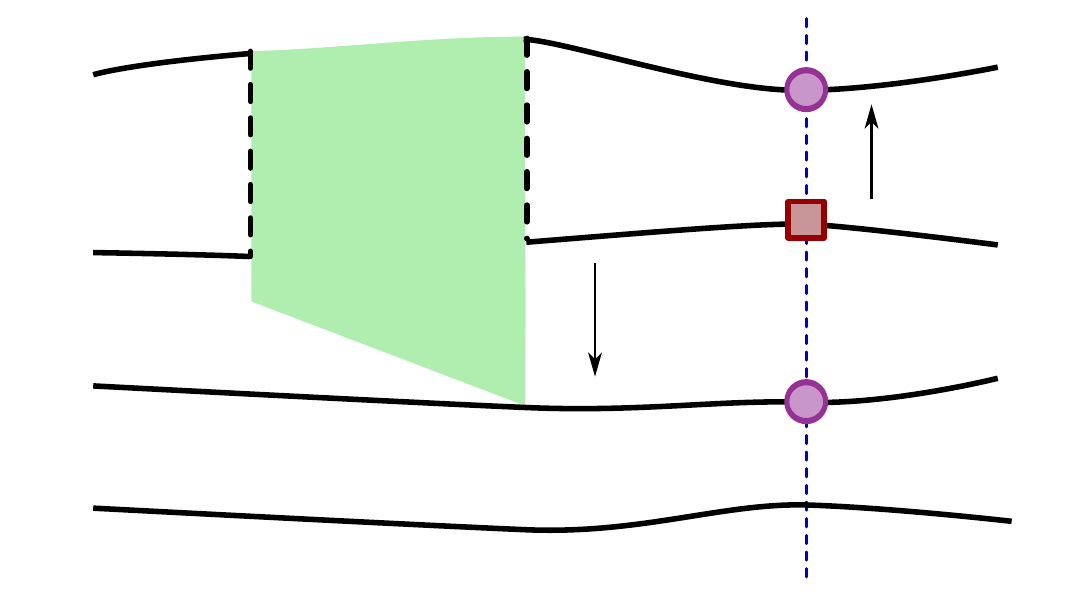 \label{fig:region_regular_canc}}
\subfigure[]{ 	\def\svgwidth{0.27\linewidth} 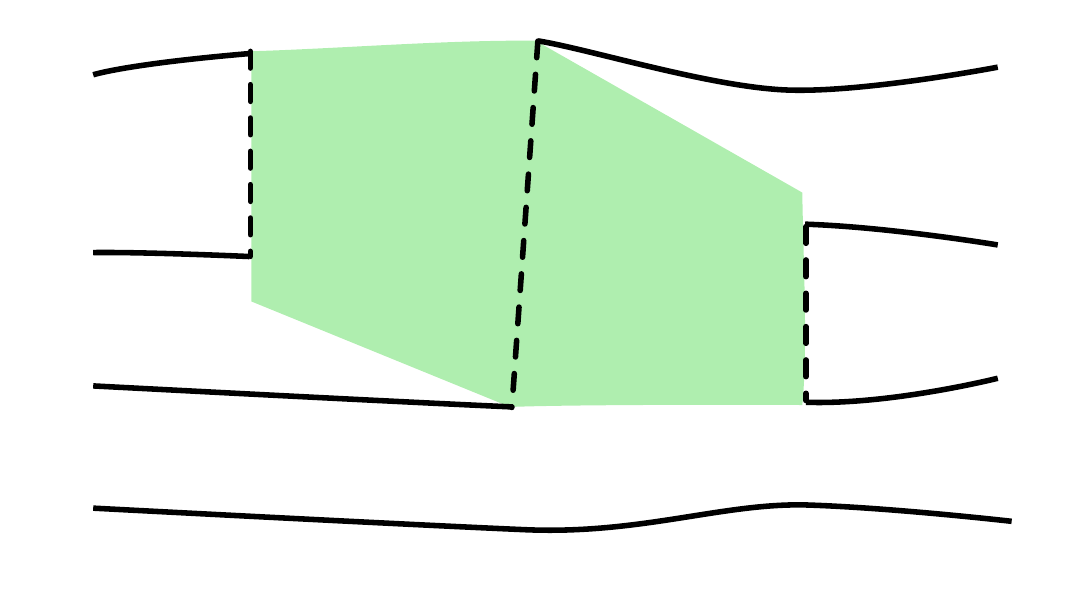\label{fig:region_regular_canc2}}
\caption{Cancellation of $R_0 = R_{(t_1, t_2)}(\alpha,\beta)$ and $R_1 =
  R_{(t_2, t_3)}(\beta,\gamma)$ with regions of influence shown in green. 
  The black dashed lines represent the points added to the Jacobi set to connect the existing loops.\label{fig:cancellation}}
\end{figure*}

By construction, Jacobi segments are paired consistently within each region. Except 
for extremal regions, which already contain the minimal number of restricted critical 
points, boundary segments of a Jacobi region $R_{(a,b)}(\alpha,\beta)$ (such that 
$\p(\alpha(t)) = \beta(t)$) can be cancelled by setting $\overline f_t(\alpha(t)) = f_t(\beta(t))$ 
for all $t \in (a,b)$.  As shown for region $R_0$ in Figure~\ref{fig:region_regular_canc}, 
this cancellation will modify $f$ only within a small neighborhood around $R_0$ still 
bounded by $g^{-1}(t_1)$ and $g^{-1}(t_2)$. We call the modified region as the {\it region of influence} 
of the corresponding cancellation and point out that it does not contain portions of $\J$ 
not part of $R_0$, and thus satisfies the consistency condition. While this creates a continuous
$\overline f$ in the region, $\overline f$ is still discontinuous at the boundary, and 
constructing a corresponding smooth $f^*$ requires a non-local change. However, consider
the cancellation of $R_1$ following the cancellation of $R_0$ as shown in 
Figure~\ref{fig:region_regular_canc2}. By construction, the region of influence of $R_1$ matches 
that of $R_0$ at their shared boundary along $g^{-1}(t_2)$, since $\beta(t_2)$ is a switch point
(where $f_t(\alpha(t_2)) = f_t(\gamma(t_2))$).  Hence, continuing the cancellation in the 
obvious manner across $g^{-1}(t_2)$ creates a valid cancellation covering the interval $(t_1,t_3)$. 
Note that this would not be possible if the cancellation of $R_0$ modified both $\alpha$ 
and $\beta$, thus we choose to modify the values of either $\alpha$ or $\beta$ (as pointed out 
in Section~\ref{sec:js:rcp}). In general, given two regular regions $R_{(t_1,t_2)}(\alpha_i,\beta_i)$ 
and $R_{(t_2,t_3)}(\alpha_j,\beta_j)$ sharing a switch point, there always exists a valid simplification 
on the interval 
$(t_1,t_3)$ which removes $\alpha_i,\beta_i,\alpha_j,\beta_j$ as well as their shared
switch and image points from $\J$.

{In order to obtain a valid $\J(f^*,g)$ consisting of closed loops, the simplification must also reconnect 
the portions of $\J(f,g)$ rendered disconnected due to the cancellations. For a continuous 
simplification, this connection can be made within a single restricted function. However, a smooth 
simplification demands modifications which can not be confined locally. For example, consider the dashed 
line in Figure~\ref{fig:region_regular_canc2} connecting segments $\alpha$ and $\gamma$ which shows such 
a transition. Without loss of generality, assume $\alpha(t)$ and $\gamma(t)$ to be maxima. The corresponding 
restricted functions in $[t_2 - \epsilon, t_2 + \epsilon]$ are shown in Figure~\ref{fig:transition}. For 
cancellation of restricted critical points, $\beta(t)$ is moved towards $\alpha(t)$ for $t < t_2$, and towards 
$\gamma(t)$ for $t > t_2$. Figure~\ref{fig:transition_canc1} shows the restricted function when the transition 
is made with in a single level set. However, to obtain a smooth transition, the simplification must also modify 
$\gamma$ in $(t_2-\epsilon,t_2)$, and $\alpha$ in $(t_2,t_2+\epsilon)$. As shown in Figure~\ref{fig:transition_canc2}, 
the maxima $\gamma(t)$ and $\alpha(t)$ in the corresponding ranges are spatially shifted towards $\beta(t)$ 
such that $\beta(t_2)$ now becomes a maxima, i.e.\ $f^*(\beta(t_2)) = f(\alpha(t_2)) = f(\gamma(t_2))$. 
Since such a transition can always be created at switch points, for simplicity in the following figures, we 
assume a smooth transition and illustrate them as vertical lines (along a single level set). Although this transition may appear to 
be a non-local and hence an invalid simplification, we remind the reader that locality is required for 
continuous simplification only. Furthermore, since the Jacobi set remains unchanged outside of $[t_1,t_3]$, 
consistency is maintained, therefore, producing a valid simplification.}

\begin{figure*}[t]
\centering
\subfigure[]{ 	\def\svgwidth{0.31\linewidth} 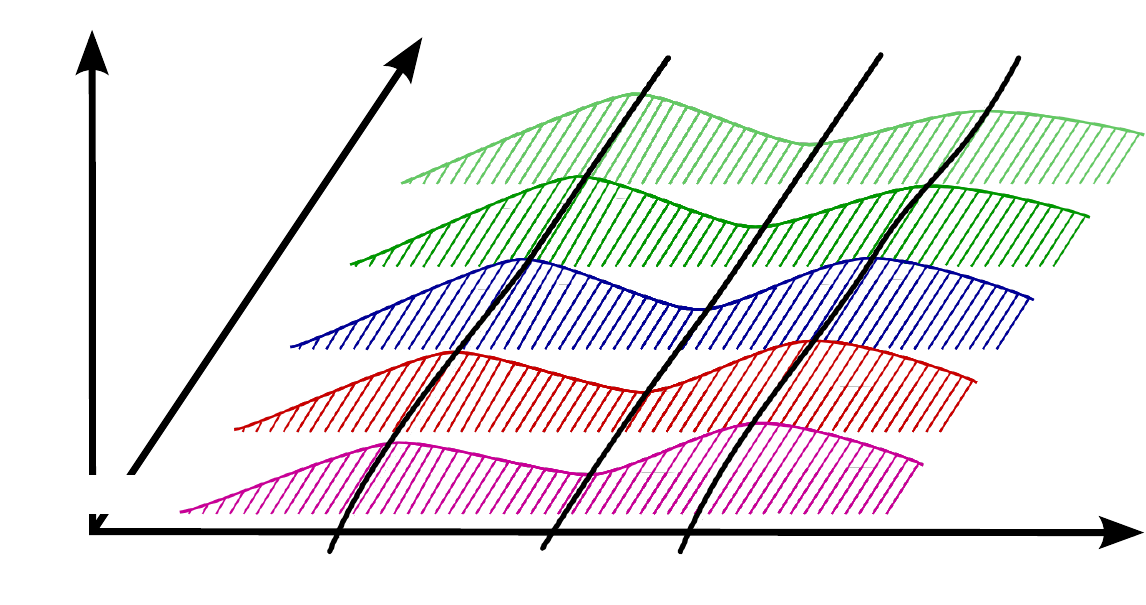 \label{fig:transition_orig}}
\subfigure[]{ 	\def\svgwidth{0.31\linewidth} 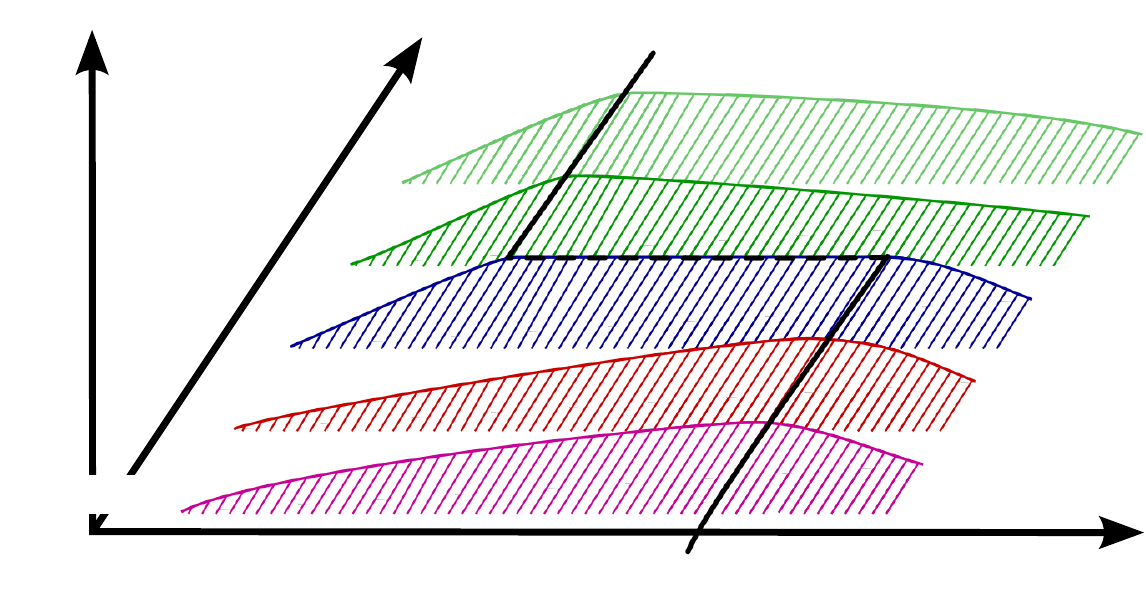 \label{fig:transition_canc1}}
\subfigure[]{ 	\def\svgwidth{0.31\linewidth} 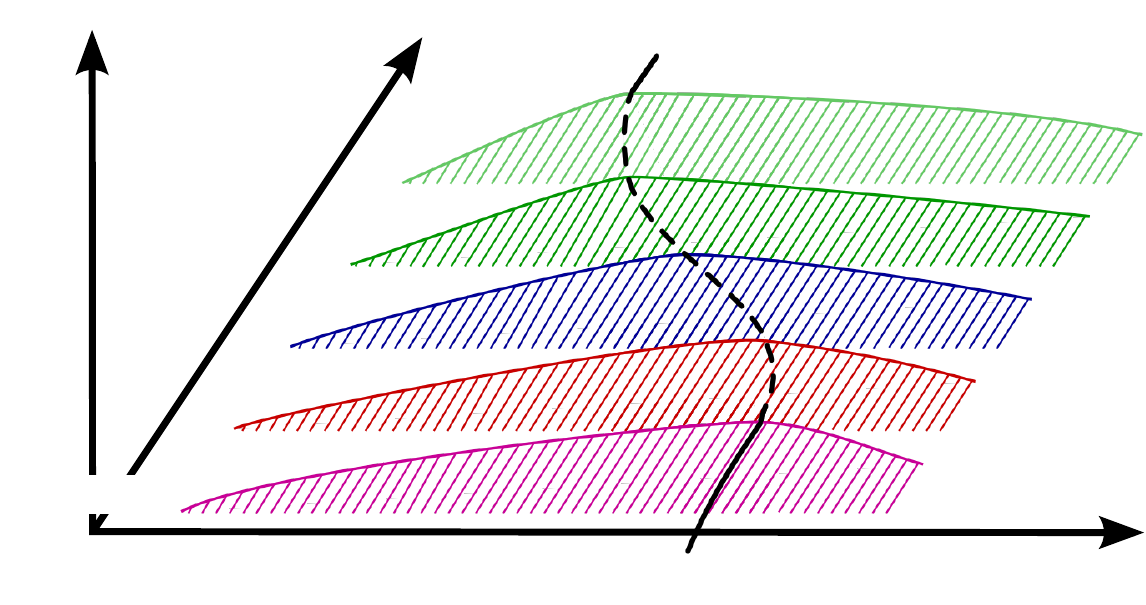 \label{fig:transition_canc2}}
\caption{The existing Jacobi curves may be connected smoothly by adding restricted critical points shown 
along the dashed curve. (a) Zoom-in view of the restricted functions in the neighborhood of \lset g {t_2} from 
Figure~\ref{fig:region_regular_orig}. (b) A continuous simplification can be obtained by creating a transition with in a single 
level set. (c) However, for a valid simplification, a smooth transition must be made by modifying the restricted functions in 
 $[t_2 - \epsilon,t_2 + \epsilon]$.\label{fig:transition}}
\end{figure*}


\subsection{Construction and cancellation of Jacobi sequences}
\label{sec:js:seq}

As discussed above, one can construct (partially) valid simplifications by simultaneously 
canceling adjacent Jacobi regions. In this section, we describe how to assemble 
\emph{Jacobi sequences} as ordered sets of regions that allow a valid simplification. 
Formally, we call two Jacobi regions \emph{adjacent} if they share a boundary point, and 
we use the function value of $g$ to induce an ordering among adjacent regions. To 
construct a sequence that admits a valid simplification, it is important to understand 
(a) where such a sequence may start or end; and (b) how to construct its corresponding 
simplification $\overline f$ and ultimately $f^*$. 

Following the discussion in Section~\ref{sec:js:regions} we claim that valid sequences are naturally
bounded by BD internal regions. This is because at the BD point, the region of influence shrinks to 
a single point and any arbitrary small interval outside the BD point allows the construction of a 
smooth $f^*$. More specifically, consider a sequence of Jacobi regions covering the interval $(a,b)$ 
that starts and ends with BD internal regions, and contains only regular regions otherwise. Given the 
discussion above, for any $\epsilon > 0$ we can create a smooth $f^*$ covering the interval 
$(a-\epsilon,b+\epsilon)$ which cancels all restricted critical points in the closure of the sequence. 
By construction $f^*$ is local, smooth, and consistent, and thus forms a valid simplification. 

Further, we note that BD external, BD side, extremal, and saddle regions can never be part of a valid simplification. 
Refer to Figure~\ref{fig:cases} and notice that it is not possible to continue across the BD point for BD external 
and BD side regions, since the discontinuity across the level set of BD point can not be removed locally. Similar 
argument holds for a saddle region, whose cancellation leaves unresolvable discontinuity around the saddle. 
Finally, an extremal region cannot be cancelled since all the level sets inside the region contain only two restricted 
critical points, and cannot be simplified further. 

As a result, valid sequences are comprised of only regular regions and BD internal regions, where they must begin 
and end with a BD internal region. Therefore, all sequences are seeded at BD internal regions and constructed 
by progression into adjacent regions monotonically in $g$ until another BD internal region is encountered, at 
which point the sequence is considered complete. Due to the ordering imposed on adjacent regions, a sequence 
cannot form loops. 


On the other hand, if during its construction, a sequence encounters any of the regions that can not be simplified,  
it is considered invalid and discarded. 
Although such regions can invalidate some sequences, this does not stop the simplification from progressing. 
If no valid sequence exists due to the presence of saddle and/or extremal regions, we perform a conventional 
$2$D critical point cancellation in $g$ to create new sequences. This cancellation does not change the Jacobi 
set structurally, and can be done independent of any sequence cancellation. Section~\ref{sec:saddle} discusses 
saddle cancellation in detail. Again, the BD external or BD side regions may invalidate some sequences. 
However, in such a case, we can always seed a new sequence from the corresponding BD internal region. 

From Section~\ref{sec:js:regions}, we know that a region $R_{(a,b)}(\alpha, \beta)$, such that 
$\p(\alpha(t)) = \beta(t)$ for all $t \in (a,b)$, can be cancelled by moving the segment $\alpha$ to the 
level of $\beta$, that is, by setting $\overline f(\alpha(t)) = f(\beta(t))$.  However, if the region is 
\emph{mutually paired}, meaning $\p(\alpha(t)) = \beta(t)$ and $\p(\beta(t)) = \alpha(t)$, one can move 
either $\alpha$ or $\beta$. This provides flexibility in sequence construction, as one can smoothly transition 
from moving $\alpha$ to moving $\beta$. Since valid simplification requires cancellation of adjacent regions 
in which the same segment can be moved to its respective partners, it follows that one can potentially cancel 
either of the two adjacent regions after canceling $R$. For example, consider Figure~\ref{fig:bidirectional} 
where regions $R_0$ and $R_1$ are already a part of a Jacobi sequence. For cancellation in $R_0$, the 
segment $\beta$ is moved to match the value of $\alpha$.  For cancellation in $R_1$, we can either continue 
moving $\beta$ towards $\gamma$, or switch segments by smoothly transitioning from moving $\beta$ to 
moving $\gamma$. The former leads to the sequence $\{R_0, R_1, R_2\}$ where $\beta$ is moved to its 
respective partners in all regions, while the latter leads to $\{R_0, R_1, R_3\}$ where $\beta$ and $\gamma$ 
are moved in $R_0$ and $R_3$ respectively, while a transition between moving $\beta$ and moving $\gamma$ 
is performed in $R_1$. 

\begin{figure*}[t]
\centering
\subfigure[]{ 	\def\svgwidth{0.26\linewidth} 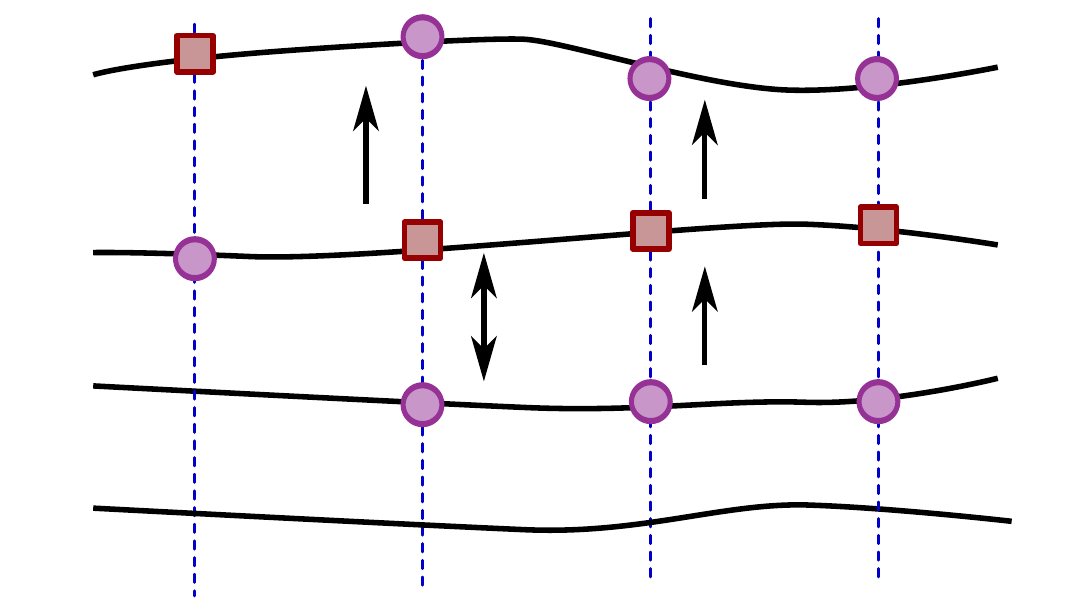 \label{fig:bidirectional1}}
\subfigure[]{ 	\def\svgwidth{0.26\linewidth} 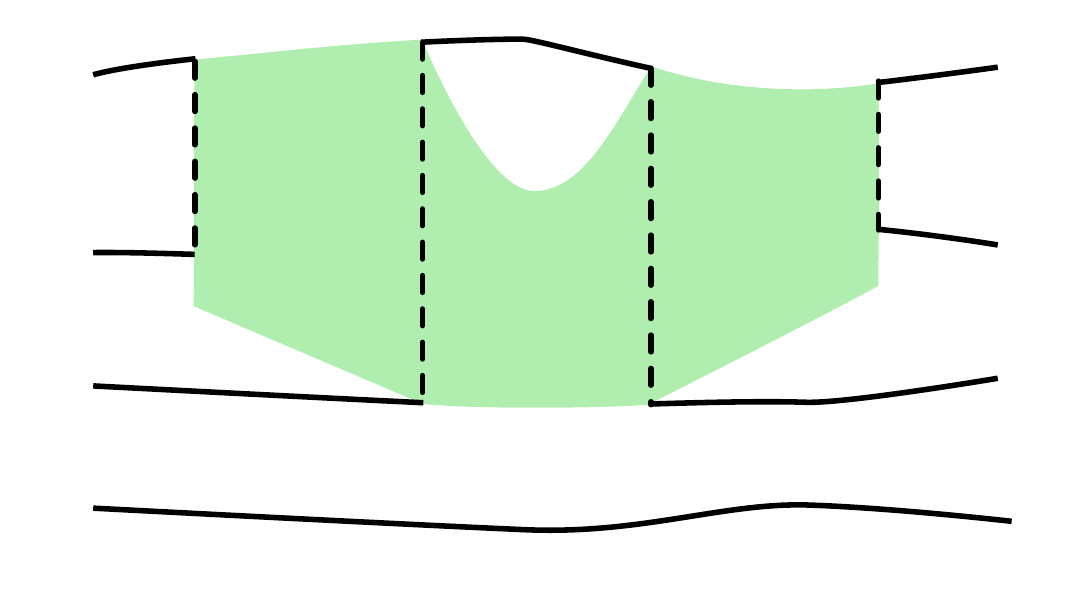 \label{fig:bidirectional2}}
\subfigure[]{ 	\def\svgwidth{0.26\linewidth} 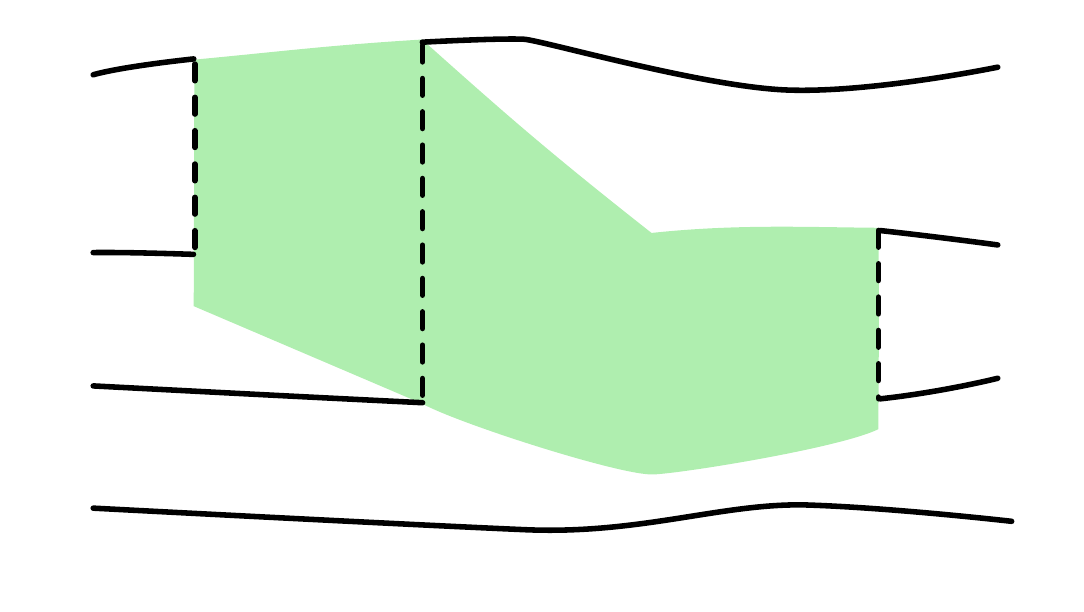 \label{fig:bidirectional3}}
\caption{Mutually-paired regions offer a choice of the segment to be moved. 
(a) Original configuration, 
where $R_0$ can be cancelled by moving $\beta$ towards $\alpha$.
Subsequently, $R_1$ can be cancelled by: (b) moving $\beta$ to $\gamma$ leading to the sequence $\{R_0, R_1, R_2\}$.
(c) smoothly transitioning between moving $\beta$ to moving $\gamma$ 
leading to sequence $\{R_0, R_1, R_3\}$.} 
 \label{fig:bidirectional}
\end{figure*}

We point out that the critical points of $f$ are naturally cancelled as critical points 
of $f_t$. Further, the construction and cancellation of Jacobi sequences can 
handle general form of structural changes to the Jacobi set.  As examples, we show the merging and 
removal of loops from Jacobi set in Figure~\ref{fig:result_loops}.

\begin{figure*}[t]
\centering
\subfigure[]{ \label{fig:loop_cancel}
  \includegraphics[width=0.23\linewidth]{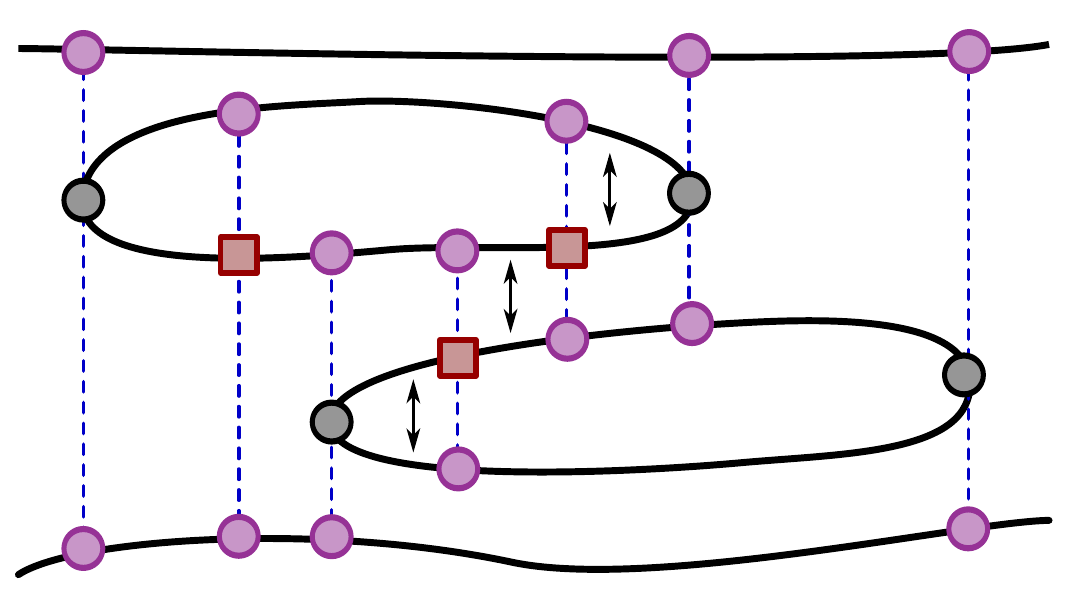}
  \includegraphics[width=0.23\linewidth]{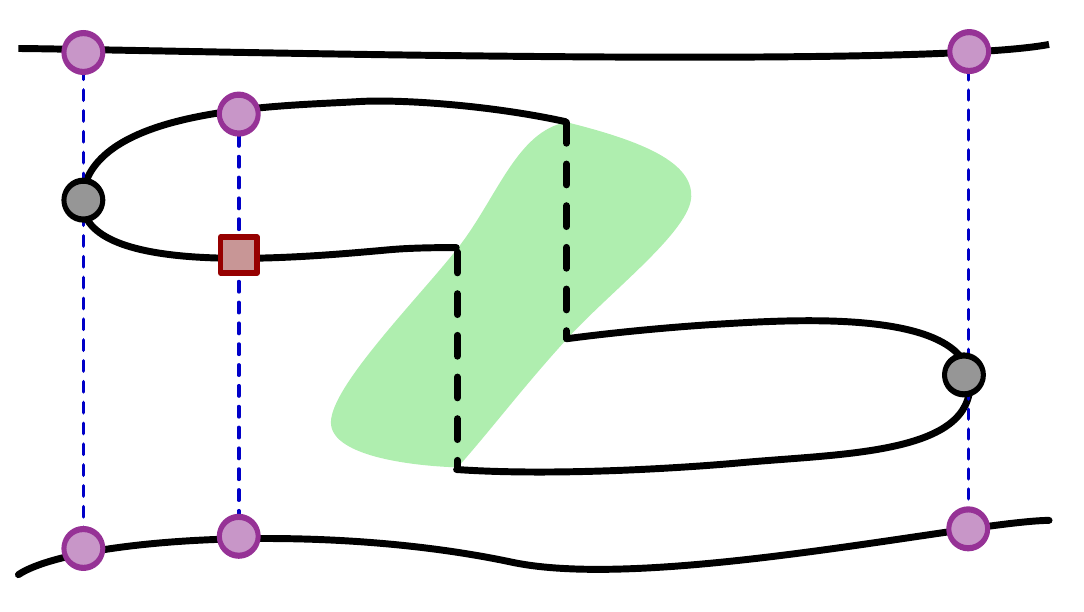}}
\subfigure[]{\label{fig:loop_cancel2}
\includegraphics[width=0.23\linewidth]{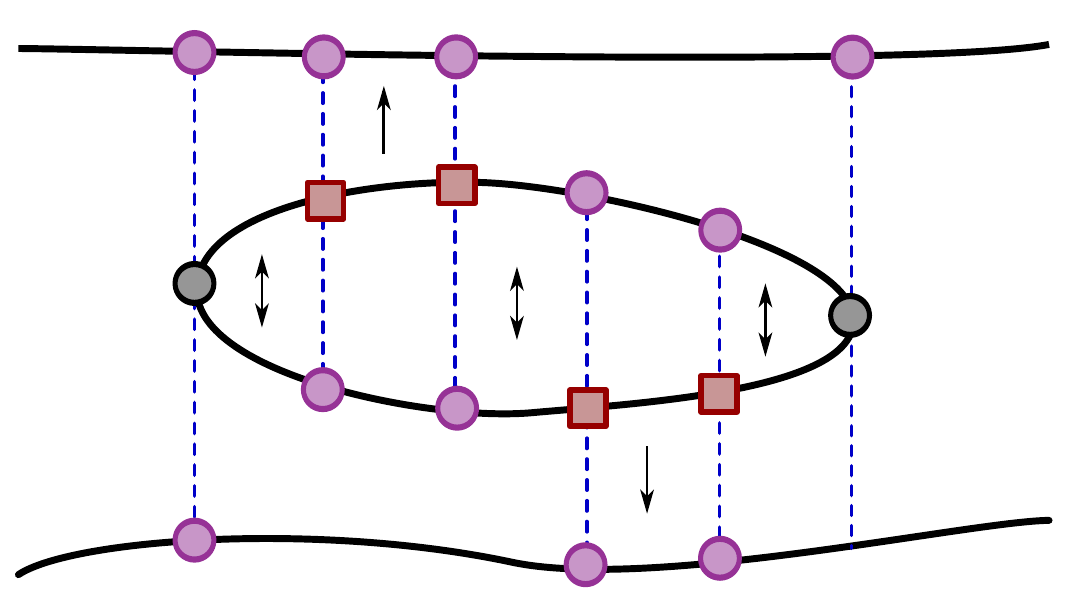}  
  \includegraphics[width=0.23\linewidth]{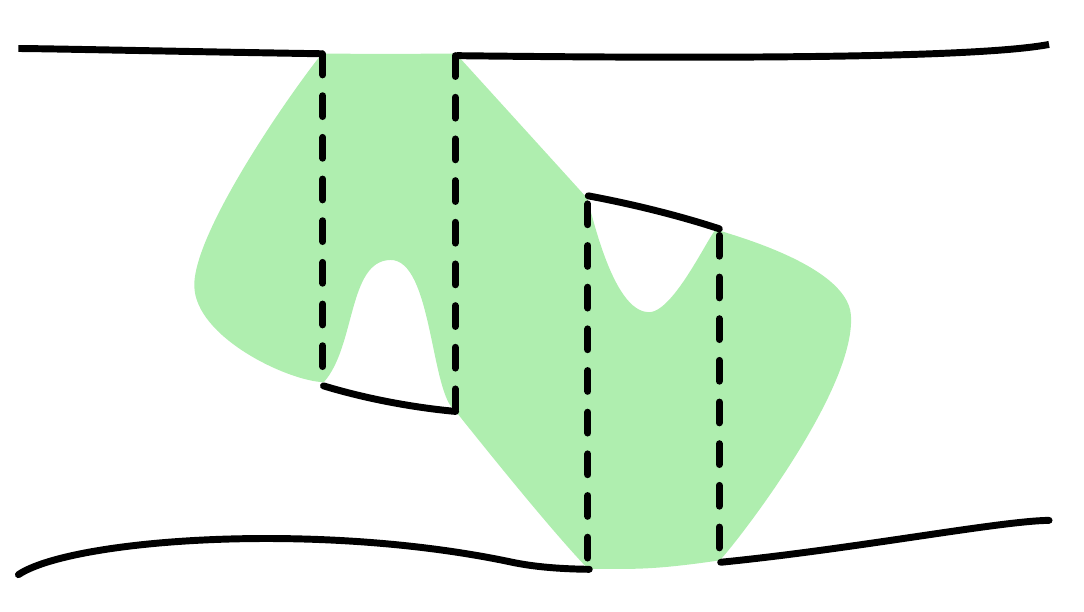}}
\caption{Our simplification algorithm performing (a) merging and 
	(b) removal of loops. Pairings are shown only for the regions that are cancelled, and the corresponding regions of influence 
	for the simplification.\label{fig:result_loops}}
\end{figure*}

\subsection{Ordering the cancellations}
\label{sec:js:ordering}

In order to obtain a hierarchy on the simplification process, we need to define a metric to measure the amount 
of modification needed for each simplification step. Although the choice of the metric is flexible, we choose a 
gradient-based metric capable of measuring the relative variation between the two functions inside a region, 
i.e.\ the \emph{comparison measure} $\kappa(R)$ (see Section~\ref{sec:fundamentals}). Our choice is inspired 
by the fact that the cancellation of a region creates a \emph{flat} $f^*$ in its interior, i.e.\ $\|\grad f^*\| \leq \epsilon$. 
An alternative formulation of $\kappa$~\cite{Vij2004} by rewriting it as an integral over the Jacobi set is 
\begin{eqnarray} \label{eq:kappa2}
\kappa(R) = \frac{1}{2\area{(R)}} \int_{v \in \Jspace}  |2f(v) - f(u)- f(w)| \cdot ||\grad{g(v)}|| \di v,
\end{eqnarray}
where, $u,w \in \n_g(v)$. 
Therefore, $\kappa(R)$ for every region $R$ can be computed by integrating over its bounding segments. 
The modification needed to cancel a Jacobi sequence, is the sum of modifications of all regions in the sequence. 
Similarly, $\kappa(R)$ can also be defined for affected regions in the critical point cancellation (see Section~\ref{sec:saddle}). 
Thus, all valid {simplification steps} can be performed in the increasing order of $\kappa$.

\paragraph{Construction of simplified function $f^*$}
Given a Jacobi sequence $\{R_i\}$ spanning the level sets in $[a,b]$, the steps to construct the simplified functions 
can be summarized below.
\begin{itemize}
\item [Step 1.] For all the regions included in the sequence, a continuous function $\overline f$ is created by canceling corresponding Jacobi segments, as guided by the pairing function. 
This modification is local and is achieved by ``flattening" the function to an appropriate value. 
\item [Step 2.] A smooth simplified function $f^*$ is created by following modifications.
\begin{itemize}\dense
\item To obtain smoothness along level sets, $\overline f$ is perturbed to induce an $\epsilon_1$-slope along the level sets, using any $\epsilon_1 > 0$ (e.g.\ see
 Figure~\ref{fig:overview_ft}).
\item To obtain smoothness across level sets while maintaining consistency, $\overline f$ is perturbed in the range \mbox{$(a-\epsilon_2,a]$} and 
			$[b,b+\epsilon_2)$, using any $\epsilon_2 > 0$ (e.g.\ see Figure~\ref{fig:valid}).			
\item To obtain smoothness in transitions between existing Jacobi curves, restricted critical points are spatially shifted in an $\epsilon_3$-neighborhood of level sets, using any 
$\epsilon_3 > 0$. (e.g.\ see Figure~\ref{fig:transition_canc2}).
\end{itemize}
\end{itemize}


\section{Cancellation of critical points in $g$}
\label{sec:saddle}

As discussed in Section~\ref{sec:js:seq}, no valid simplification sequence of $f$
can cancel a critical point of $g$. However, there may exist configurations
such that all Jacobi sequences of $f$ contain critical points of $g$ and all
sequences in $g$ contain critical points in $f$. In this case, there exists no
valid sequence and the Jacobi set cannot be simplified through a standard
cancellation. Instead, we use traditional critical point cancellations to remove
pairs of critical points from either function. In this section, we show that
critical points on $g$ can be cancelled with minimal impact to the geometry of
$\J(f,g)$. Furthermore, we describe the change in $\kappa$ caused by such a
cancellation, and how the pairing among Jacobi segments is affected. Rather than
the typical persistence pairing, we use a slightly relaxed notion of critical
point pairing centered around the notion of an {\it isolated} pair.

\begin{definition}
  A saddle-maximum $(\Sd,\Mx)$ pair of a function $g$ is called \emph {isolated} if
  the component of the super-level set \lsetgs\ containing $\Mx$ does not contain any other
  critical points of $g$.
\end{definition}

Clearly, on a simply-connected domain all extrema of $g$ except for a single
maximum and minimum can be removed through successive cancellation of isolated
pairs. The section is divided into four parts: First, we describe how given an
isolated maximum-saddle pair $(\Sd,\Mx_1)$ in $g$, one can construct a smooth
function $g^*$ that only differs from $g$ in a arbitrary small neighborhood of
the super-level set around $\Mx_1$; Second, we show that $\J(f,g) = \J(f,g^*)$
except for an $\epsilon$-neighborhood around $\Sd$; and Third, we discuss how the
cancellation affects Jacobi segments and regions; and last, we identify  the modification 
needed for this cancellation.

\begin{figure*}[htbp]
\centering
\subfigure[]{		
\def\svgwidth{0.6\linewidth} 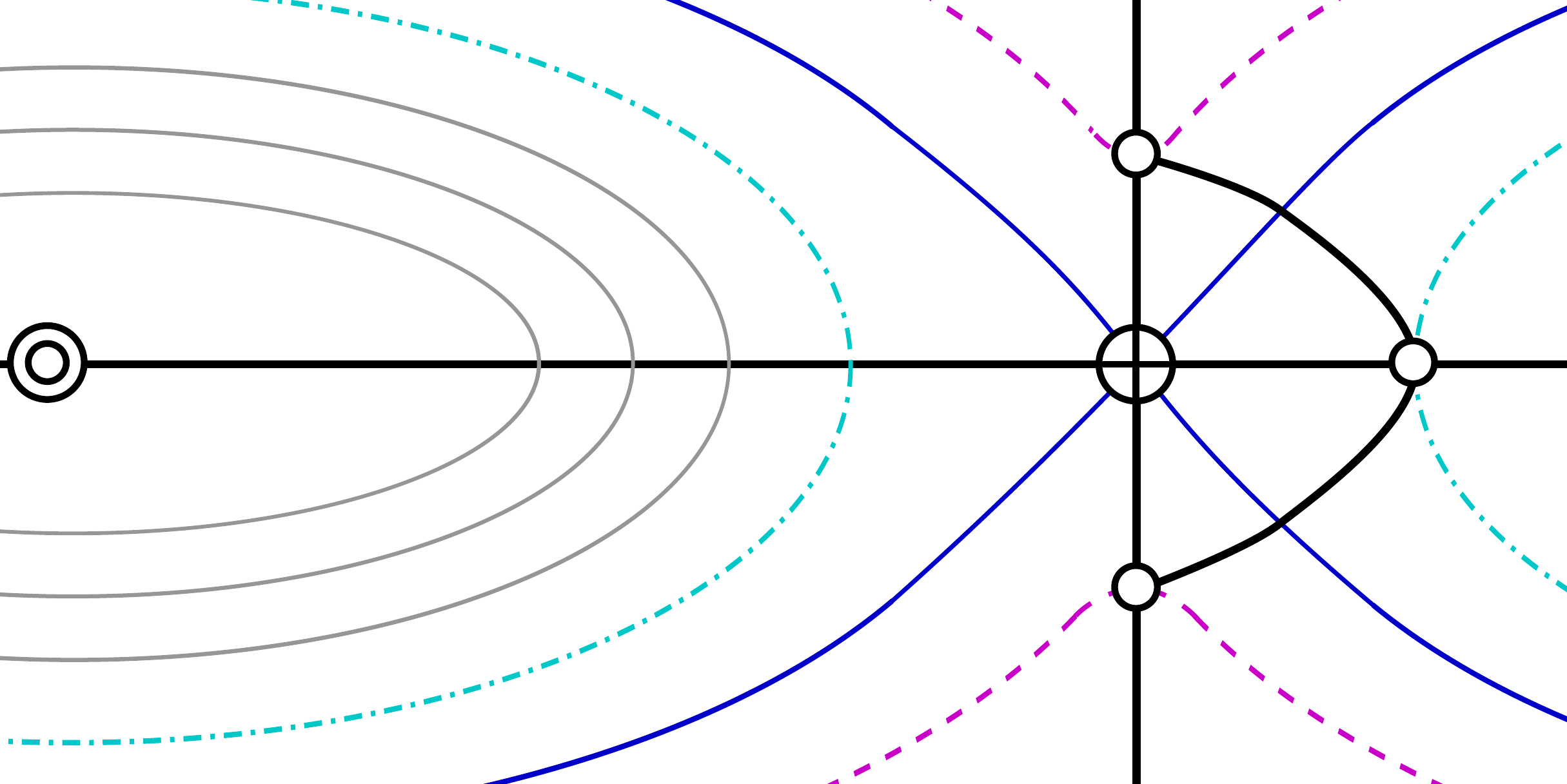 
\def\svgwidth{0.6\linewidth} 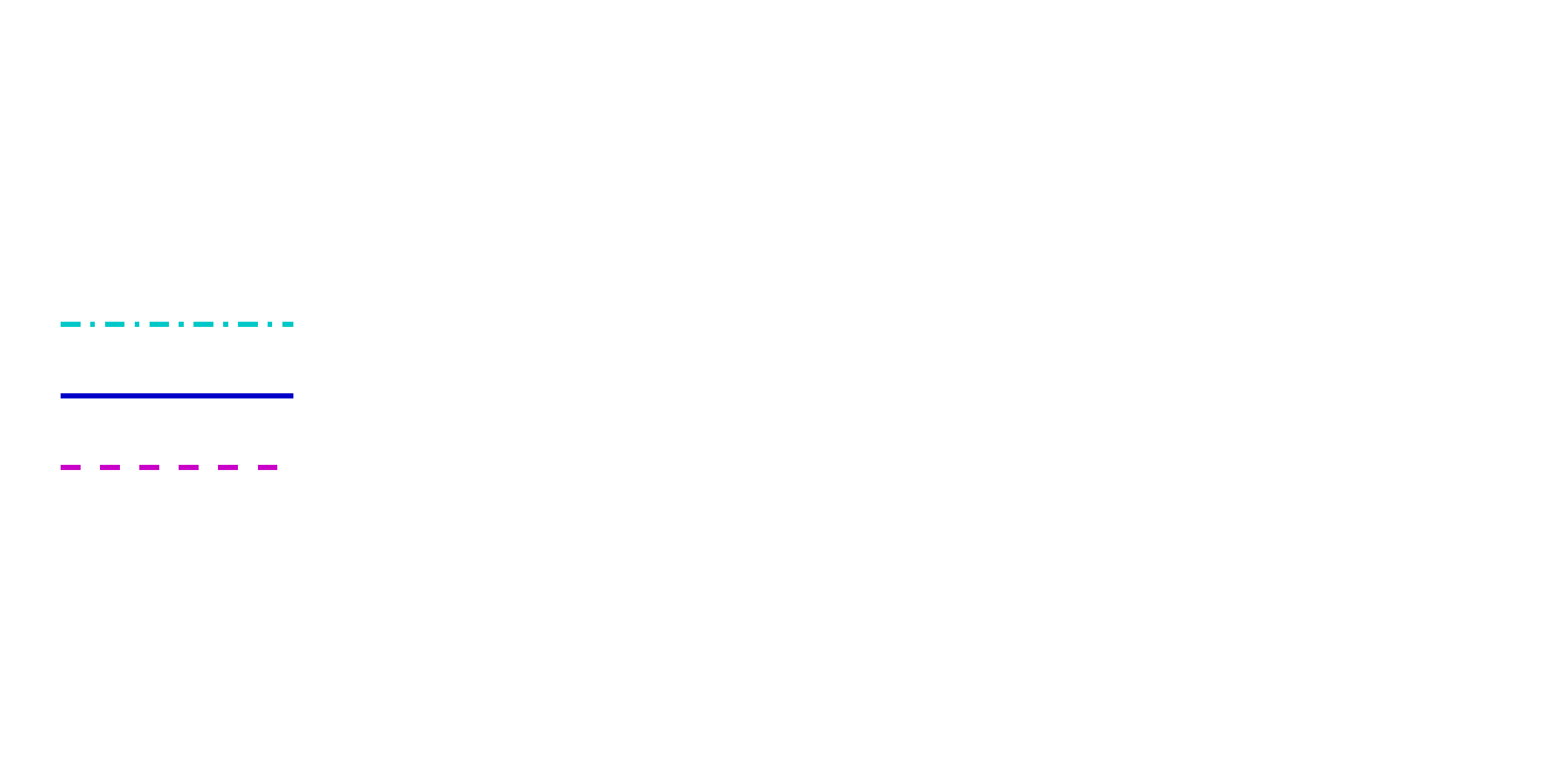  \label{fig:sadmax_zoom1}
\hspace{-15em}}

\subfigure[]{		\includegraphics[width=0.45\linewidth]{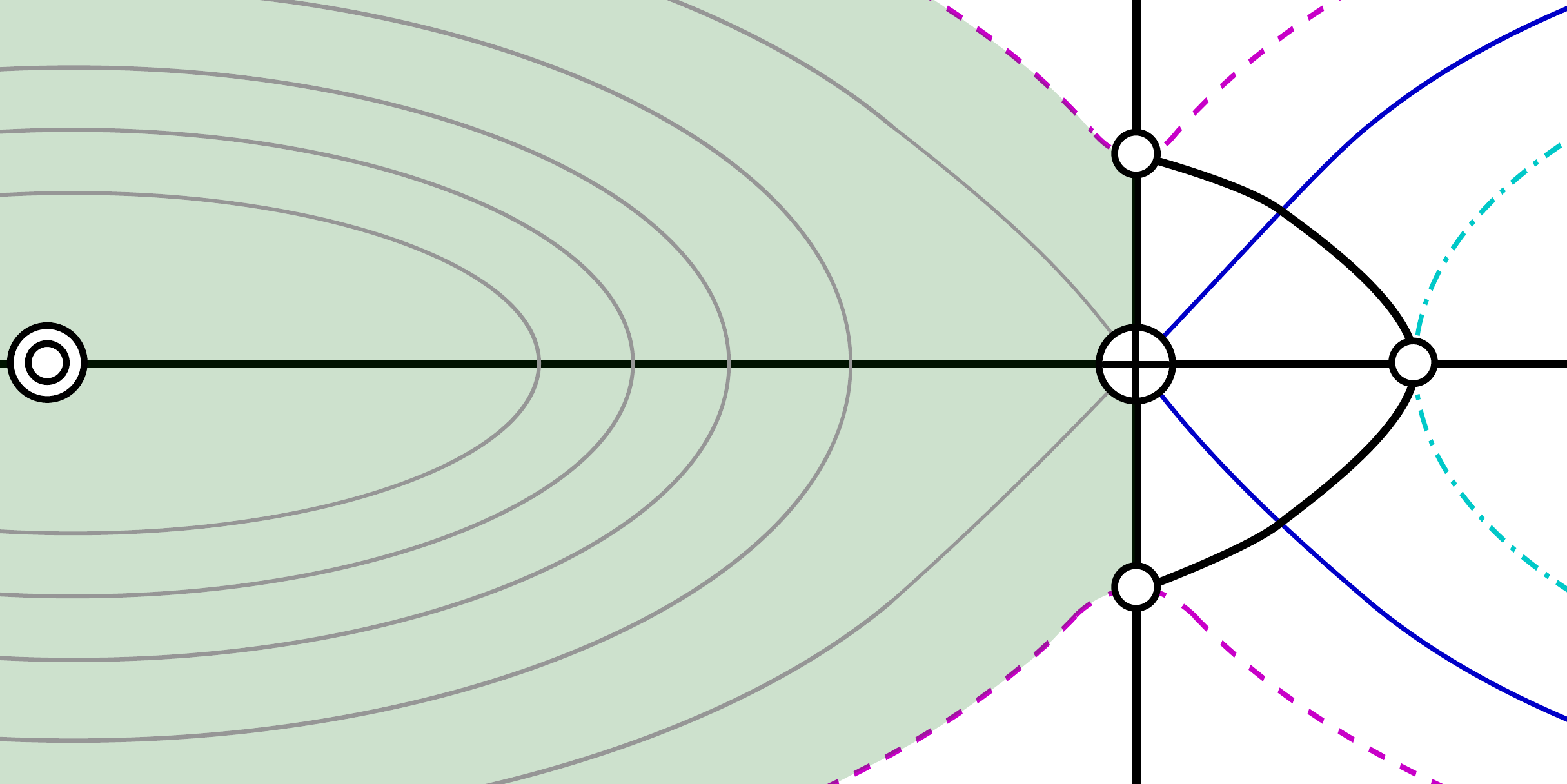} \label{fig:sadmax_zoom2}}
\hspace{2em}
\subfigure[]{		\includegraphics[width=0.45\linewidth]{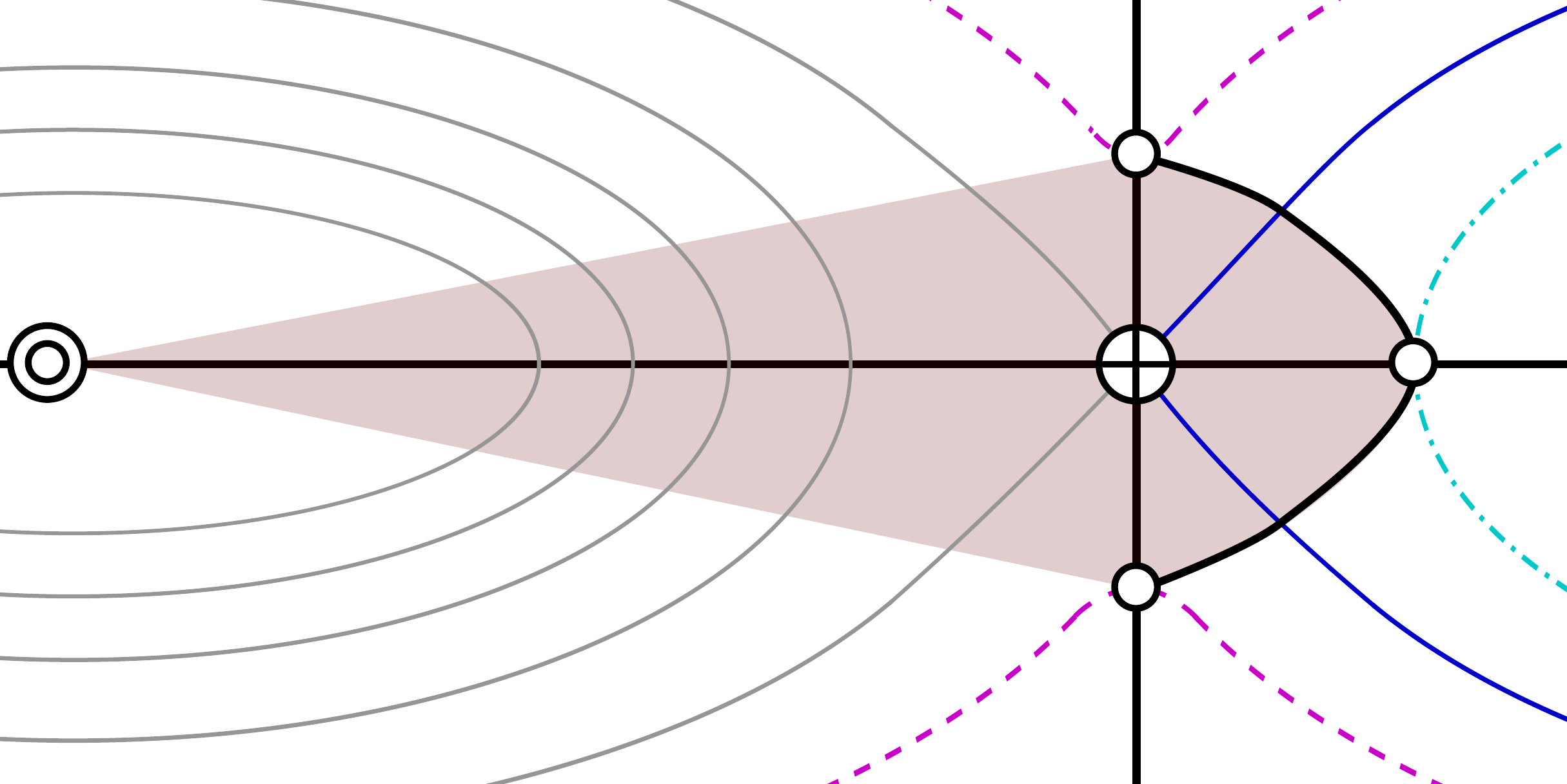} \label{fig:sadmax_zoom3}}

\subfigure[]{		\includegraphics[width=0.6\linewidth]{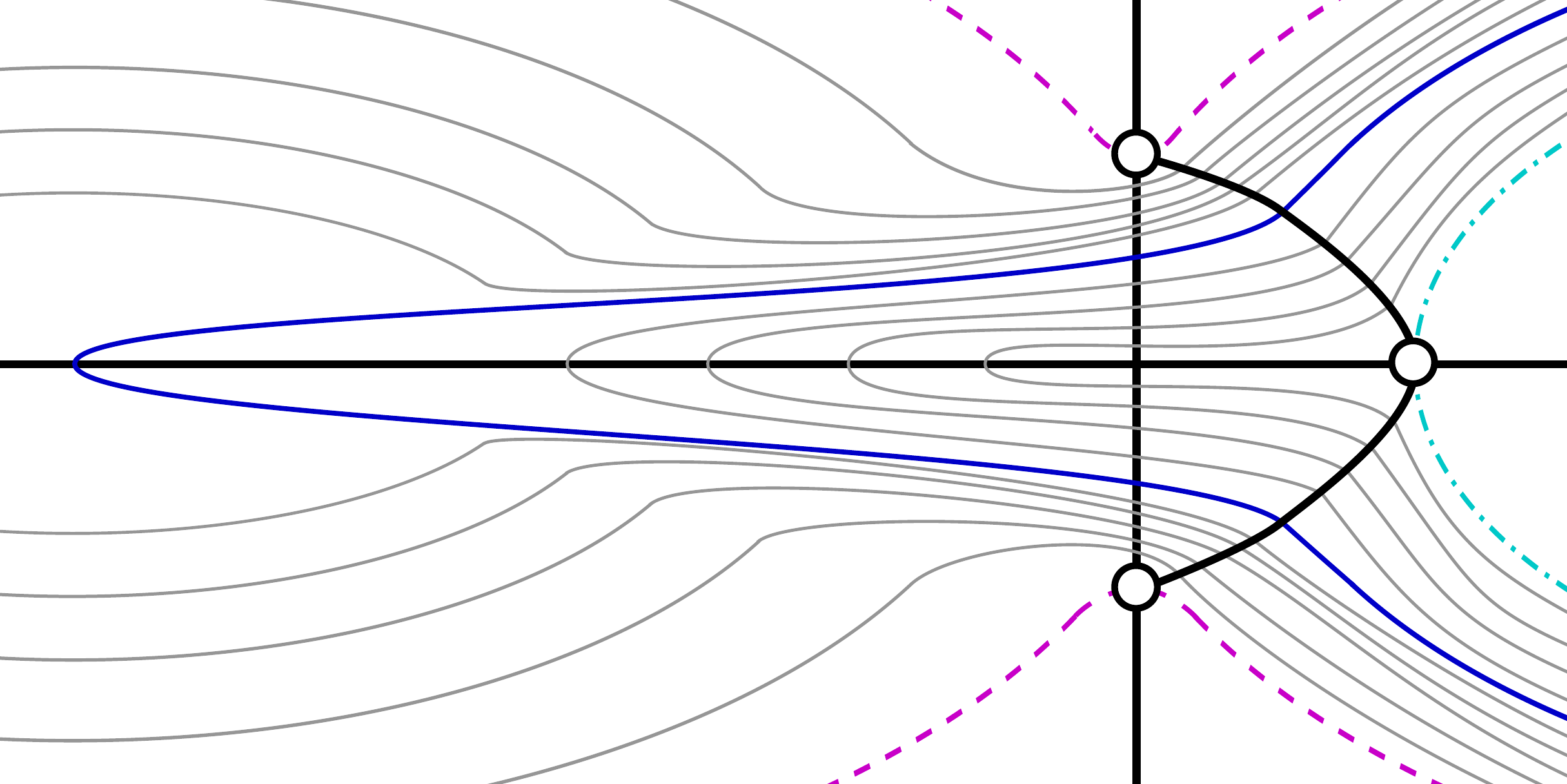} \label{fig:sadmax_zoom4}}

\caption{Cancellation of a saddle-maximum pair, $(\Sd,\Mx_1)$. 
(a) Level sets of the original function, $g$. Line $l_1$ connects $\Mx_1$ with $a \left(\in \lsetgsp\right)$ 
through $\Sd$. Line $l_2$ connects $b_1$ and $b_2$ with $b_i \in \lsetgsn$. Line $l_3$ connects $b_1$, 
$a$, and $b_2$. 
(b) The shaded region bounded by \lsetgsn\ and $l_2$ is rescaled to the range $[g(\Sd)-\epsilon,g(\Sd)]$. 
This creates a discontinuity along $l_2$. 
(c) The shaded region $C$ is identified and $g$ is modified such that $\grad g^*$ is anti-parallel to $\grad g$  
 along $l_1$, and $\grad \gs(x) = \grad g(x)$ for all $x \in \partial C$. This cancels the $(\Sd, \Mx_1)$ pair, and 
removes the discontinuity in~\ref{fig:sadmax_zoom2}. 
(d) The final function, $\gs$ does not contain the saddle $\Sd$, and the maximum $\Mx_1$.
The modification (from $g$ to $g^*$)  is confined only to the shaded regions in Figures~\ref{fig:sadmax_zoom2} 
and~\ref{fig:sadmax_zoom3}.
\label{fig:saddlecancellation_zoom}}
\end{figure*}

\newparagraph{Critical points pair cancellation}
Consider the canonical level set structure of $g$ around $\Sd$ shown in
Figure~\ref{fig:sadmax_zoom1} and the three level sets 
\lsetgs, \lsetgsp, and \lsetgsn. 
Furthermore, consider
three lines: $l_1$, connecting $\Mx_1$ with $\Sd$ and a point $a$ on 
\lsetgsp; 
$l_2$, connecting $\Sd$ with two points $b_1$ and $b_2$ on 
\lsetgsn; 
and $l_3$ connecting $b_1$, $a$, $b_2$ as shown in Figure~\ref{fig:sadmax_zoom1}.

To create a $g^*$ that cancels $(\Sd,\Mx_1)$ we first rescale $g$ within the
super-level set of $g^{-1}(g(\Sd)-\epsilon)$ on the left of $l_2$
(Figure~\ref{fig:sadmax_zoom2}) by monotonically mapping the range
$[g(\Sd)-\epsilon,g(\Mx_1)]$ to the range $[g(\Sd)-\epsilon,g(\Sd)]$. Choosing
the appropriate map, the resulting function $\bar{g}$ is smooth except at $l_2$
where it is discontinuous. Note that $\bar{g}(\Mx_1) = g(\Sd)$. Subsequently, we
modify $\bar{g}$ to be monotonically increasing along $l_1$. Finally, we define
a cone $C$ around $l_1$ as a region enclosed by $l_3$ and containing $\Mx_1$
(Figure~\ref{fig:sadmax_zoom3}). Now we define a smooth $g^*$ with $g^* =
\bar{g}$ outside $C$, $\grad g^*(x) = \grad g(x)$ for all $x \in \partial C$,
and $\grad g^*$ is anti-parallel to $\grad g$ 
for all $x \in l_1$, as shown in
Figure~\ref{fig:sadmax_zoom4}. Such a $g^*$ exists for all $\epsilon > 0$ as the
solution of a boundary value problem.

\newparagraph{Jacobi set geometry}
\begin{figure*}[ht]
\centering
\subfigure[]{		\def\svgwidth{0.32\linewidth} 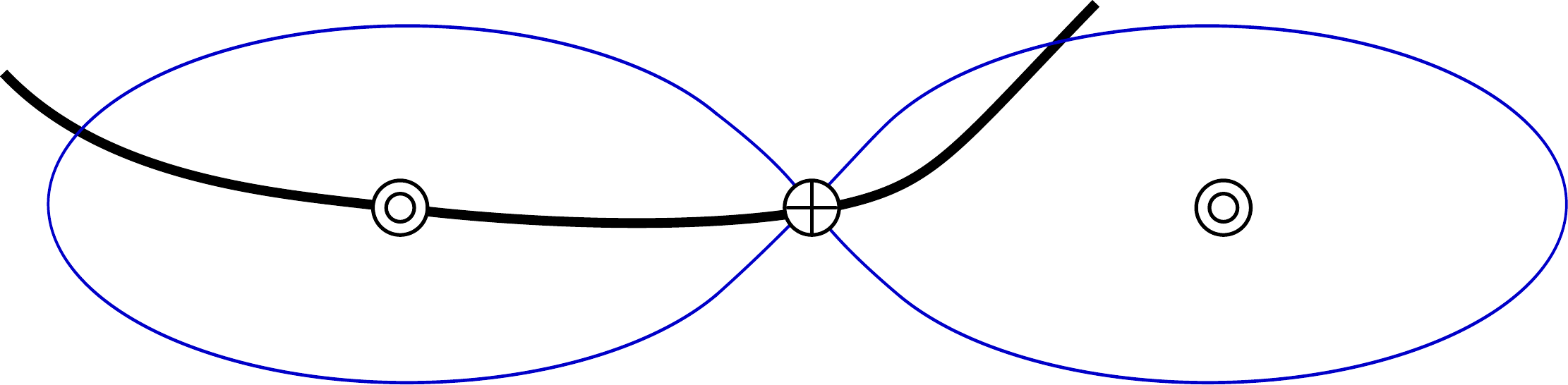  \label{fig:sad_case1}}
\hspace{-1em}
\subfigure[]{		\def\svgwidth{0.32\linewidth} 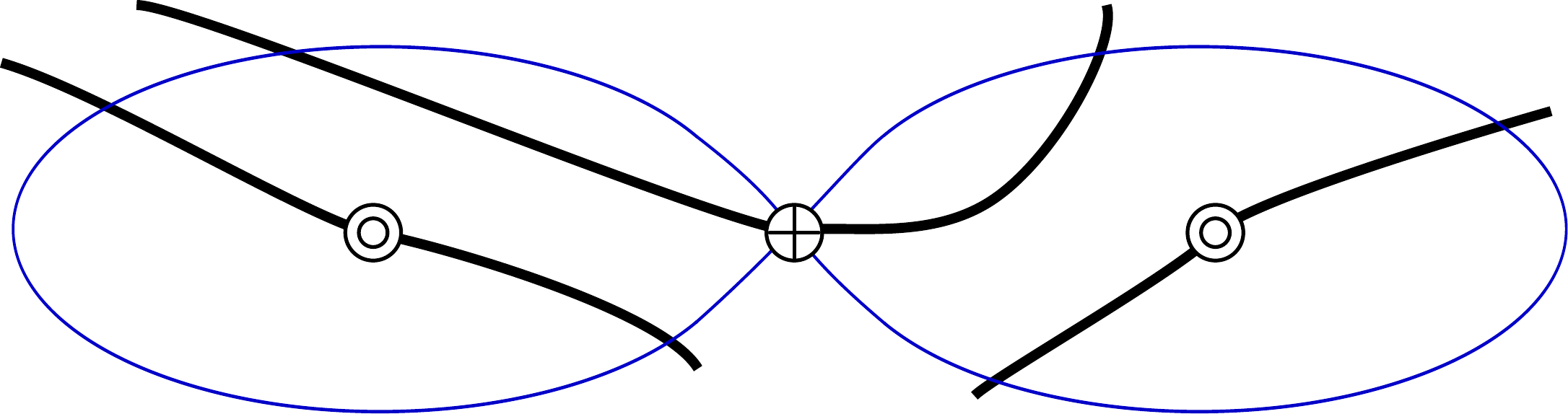  \label{fig:sad_case2}}
\hspace{-1em}
\subfigure[]{		\def\svgwidth{0.32\linewidth} 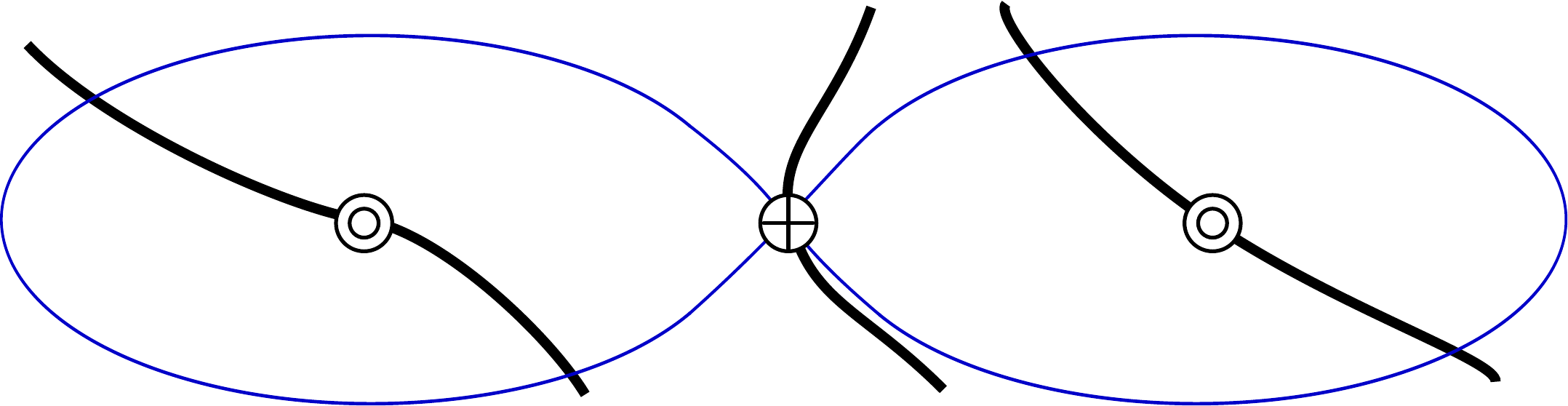  \label{fig:sad_case3}}
\caption{Different cases of Jacobi set connectivity for an isolated saddle-maximum pair $(\Sd,\Mx_1)$. 
The saddle and the maximum may be parts of by (a) the same Jacobi set component $\J_{\Sd,\Mx_1}$, or 
(b) and (c) separate Jacobi components $\J_\Sd$ and $\J_{\Mx_1}$ respectively. 
$\Lspace_1$ and $\Lspace_2$ are super level sets of $g$ surrounding $\Mx_1$ and $\Mx_2$ respectively. 
\label{fig:saddle_cases}}
\end{figure*}
Generically, there exist three different configurations of Jacobi sets in the neighborhood 
of an isolated critical point pair as shown in Figure~\ref{fig:saddle_cases}. The most common 
configuration is a Jacobi set connecting $\Mx_1$ and $\Sd$ (Figure~\ref{fig:sad_case1}). In 
this case we can define $l_1 \subset \J$ and $l_2$ as a subset of the unstable $1$-manifold  
of $\Sd$ (i.e.\ gradient descending path of a saddle) which guarantees that $\grad g^*(x)$ 
and $ \grad g(x)$ are aligned, for all $x \in \J(f,g)$. It follows that $x \in \J(f,g)$ implies 
$x \in \J(f,g^*)$. Furthermore, for all $x \notin C$ we have $\grad g^*(x) = \grad g(x)$ which 
implies $\J(f,g) = \J(f,g^*)$ for $\Mspace \setminus C$.  However, with $g^*$ as defined above, 
there may exist additional points $x \in C$ with $x \in \J(f,g^*)$. By construction, these must 
be part of isolated Jacobi components (loops) entirely contained in $C$. As such they must form 
a valid cancellation sequence and can be remove using the approach discussed in Section~\ref{sec:jssimp}. 

The situations shown in Figures~\ref{fig:sad_case2} and~\ref{fig:sad_case3} follow a similar 
argument except that it cannot be guaranteed that $\J(f,g) = \J(f,g^*)$ around $\Sd$. However, 
since this portion of the Jacobi set enters and exists $C$ exactly once, there must exist a $g^*$ 
that connects the entry and exit points with a single line of the Jacobi set containing no BD points. 
Therefore, $\J(f,g) \neq \J(f,g^*)$ only in a small neighborhood around $\Sd$. 

\newparagraph{Modifications in Jacobi segments and Jacobi regions}
To understand how the Jacobi segments and Jacobi regions are affected by this cancellation, we refer to Figure~\ref{fig:saddlecancellation_full}. 
In addition to the lines and points described above (in Figure~\ref{fig:saddlecancellation_zoom}), 
let $c_1$ and $c_2$ be the points where $l_3$ intersects with \lsetgs, and $d$ the point of intersection of $l_1$ with \lsetgsn. 
The cancellation modifies the shape of the level sets in specific ways shown in Figure~\ref{fig:sadmax_zoom4}. 
That is, each point $x \in \overline{d\Mx_1}$ (i.e. the line between $d$ and $\Mx_1$) is connected to some $y \in \overline{b_1 c_1}$ 
on one side (of $g^{-1}(\Sd)$), and some $z \in \overline{b_2 c_2}$ on the other side. 
Similarly, each $x \in \overline{\Mx_1 a}$ is connected either to some $y \in \overline{c_1 a}$, or some $z \in \overline{c_2 a}$. 
The pairings in $f_t$ must be recomputed along the modified level sets, and new regions need to be created. 

To illustrate the modifications in pairings and Jacobi regions, we give an example of such cancellation in Figure~\ref{fig:saddle_example}.
The figure shows that most of the regions are unaffected. Only the regions that included the cone $C$ as described above are modified, 
and extend along the new level sets.

\begin{figure*}[!ht]
\centering
\subfigure[]{           \def\svgwidth{0.7\linewidth} 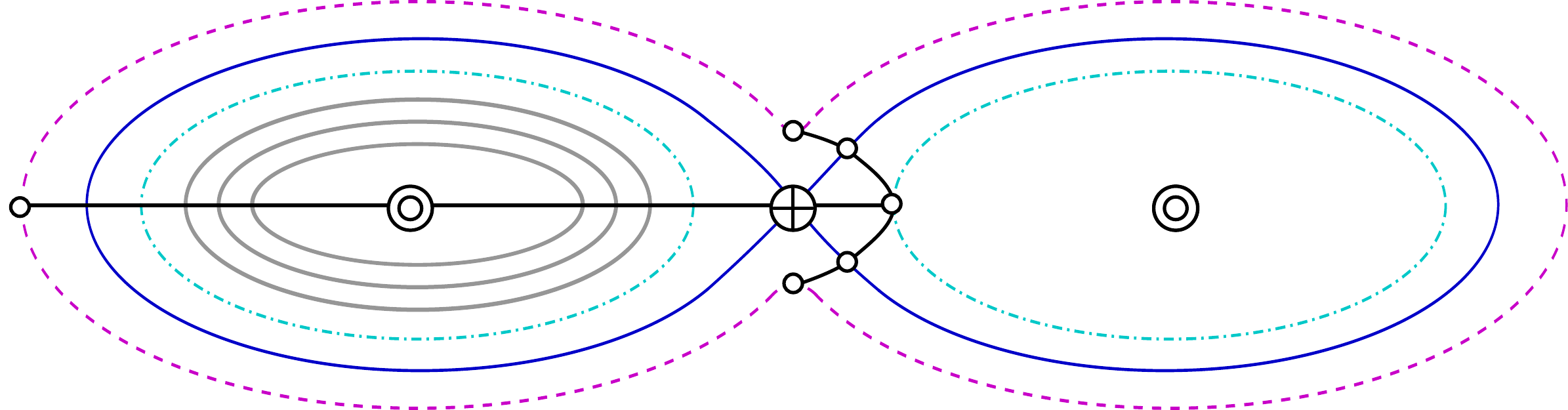 \label{fig:sadmax_full1}}
\subfigure[]{           \def\svgwidth{0.7\linewidth} 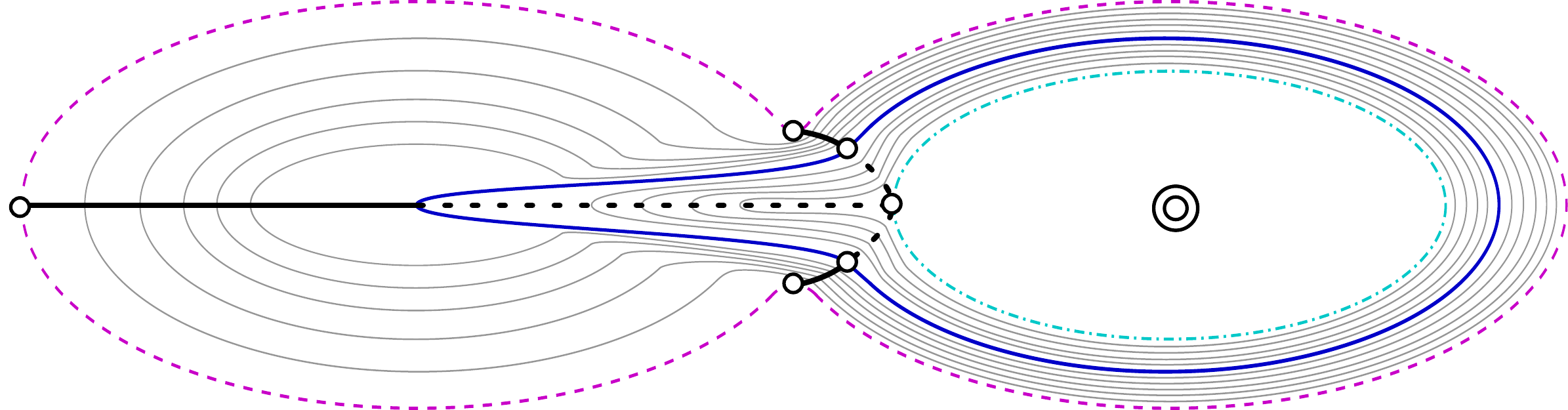 \label{fig:sadmax_full2}}
\caption{The modifications in level sets due to the cancellation of $(\Sd,\Mx_1)$, as shown (a) before and (b) after the cancellation. 
\label{fig:saddlecancellation_full}}
\end{figure*}

\begin{figure*}[!ht]
\centering
\subfigure[]{           \def\svgwidth{0.7\linewidth} 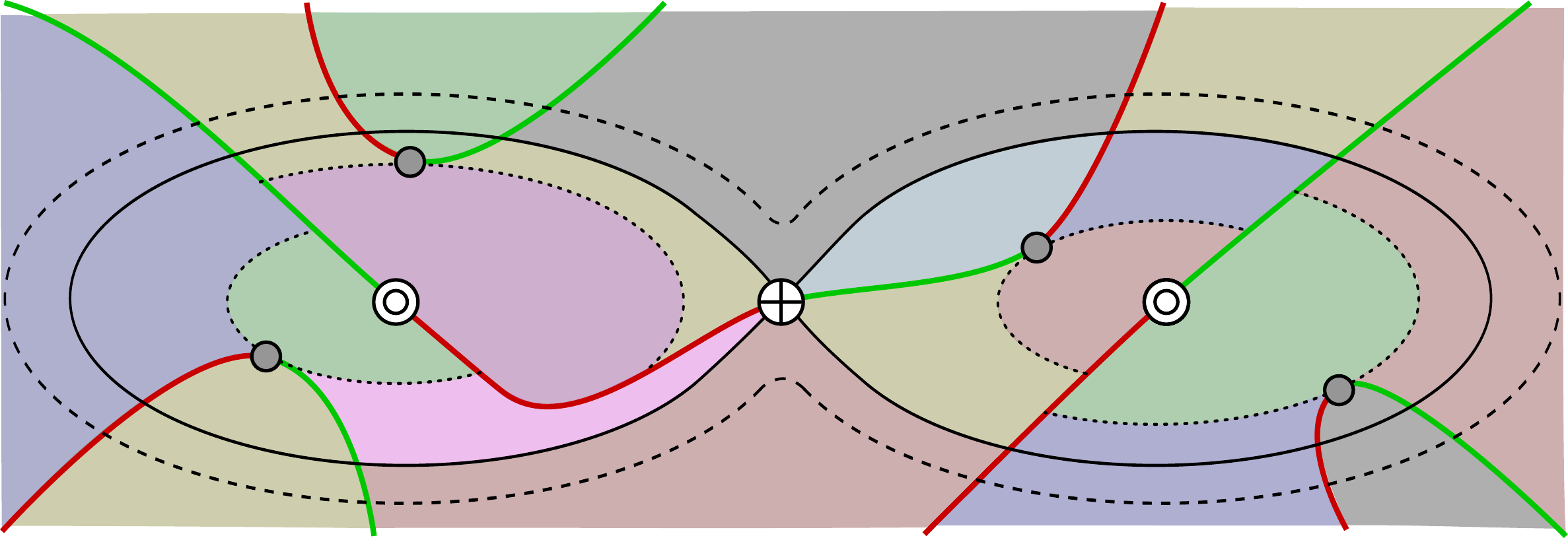 \label{fig:sadmax_ex1}}
\subfigure[]{           \def\svgwidth{0.7\linewidth} 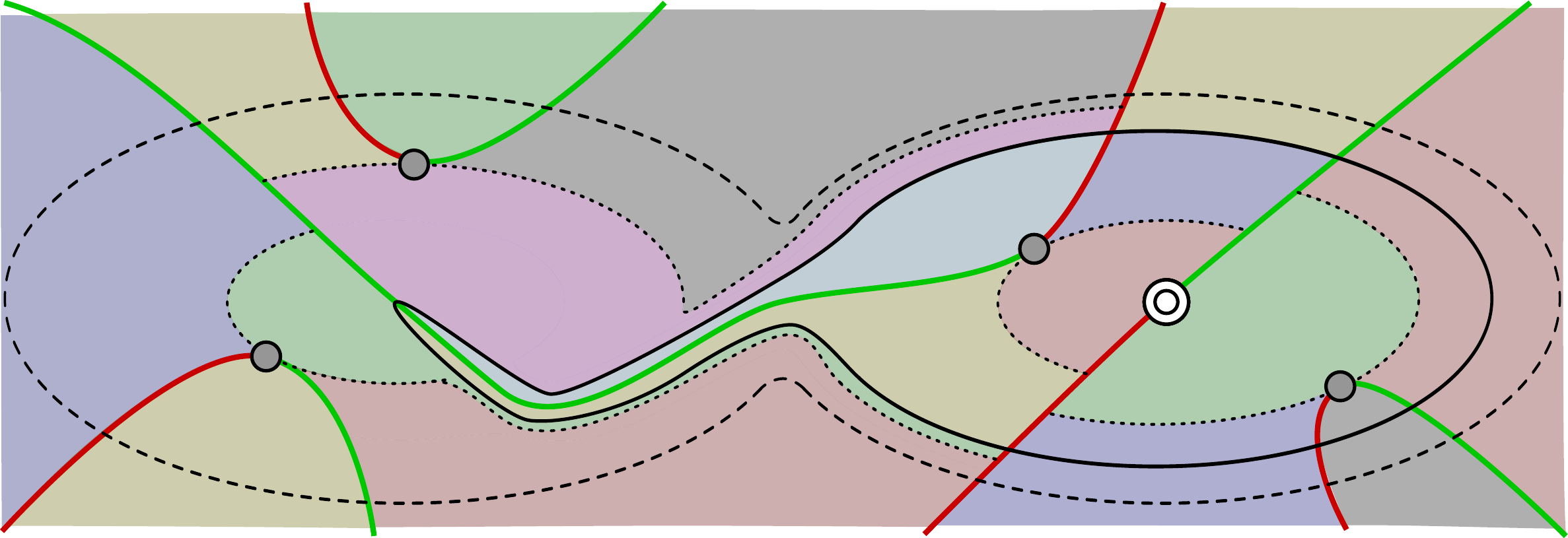 \label{fig:sadmax_ex2}}
\caption{Effect of saddle cancellation on Jacobi segments and Jacobi regions. In addition to the level sets of BD points (dotted), level sets \lsetgs\ (solid) and \lsetgsn\ (dashed) are shown for reference. 
The Jacobi set is shown as red-green curve, with color representing the criticality (i.e. restricted maximum or minimum).  
Jacobi regions corresponding to (a) the original Jacobi set $\J(f,g)$, and (b) the Jacobi set after cancellation, $\J(f,g^*)$, 
are shown in different colors. 
\label{fig:saddle_example}}
\end{figure*}

\newparagraph{Modification needed for the cancellation}
We  note that 
$|\nabla g^*(x)| = O(\epsilon)$ for all $x \in \Lspace_1$ (where $\Lspace_1$ is the super level set surrounding $\Mx_1$). 
Then, the comparison measure of $\Lspace_1$ after cancellation, $\kappa^*$, is given by
\begin{align*}
	\kappa^*(\Lspace_1) &= \frac{\int_{\Lspace_1} \|\nabla f(x) \times \nabla \gs(x)\| \di x}{\area(\Lspace_1)} 
							  = O(\epsilon)
\end{align*}
Note that $\kappa^*$ is independent of both the difference in the function values of $\Sd$ and $\Mx_1$, 
and the shape of $\Lspace_1$.
Thus, the amount of perturbation introduced by this cancellation is 
approximately $\lim_{\epsilon \to 0} (\kappa - \kappa^*) = \kappa$. 
 

\section{Summary and Correctness}
\label{sec:summary} 

Given the discussion on the cancellation of restricted critical points of $f_t$, and 
critical points of $g$, we now summarize the complete procedure to simplify a
given Jacobi set. So far, all the discussion has focused on modifying $f$ with 
respect to the level sets of $g$. However, we may wish to interleave the
modifications of either of these functions with respect to the other.  Thus, to
simplify the Jacobi set, we need to identify all Jacobi sequences with respect to
both -- the level sets of $f$ and the level sets of $g$.
\begin{itemize}
\item[Step 1.] Identify all possible simplification steps with respect to the level sets
  of $g$/$f$, by creating all possible Jacobi sequences, and identifying all 
  isolated saddle-extremum pairs.  
\begin{itemize}\dense
\item Compute the pairings between restricted critical points and identify the switch points.
\item Create Jacobi segments by decomposing $\J$ into subsets bound by the BD points in 
		$\J$, critical points of $g$/$f$, switch points, and their images. 
\item Create Jacobi regions using the pairings induced on the segments, and compute their 
		$\kappa$. 
\item Create Jacobi sequences $\{S\}_g$ and $\{S\}_f$ by seeding them at BD internal regions, 
		and propagating monotonically into adjacent regions in a depth-first manner, and 
		compute its $\kappa$.  
\item Identify all isolated saddle-extremum pairs $\{P\}_g$ and $\{P\}_f$, and compute their 
		$\kappa$. 
\end{itemize}
\item [Step 2.] Store all sequences $\{S\}_f$ and $\{S\}_g$, and all pairs of saddle-extremum pairs 
	$\{P\}_f$ and $\{P\}_g$ into a common list $\Lcal$, ordered by their $\kappa$, the 
	amount of modification needed for their cancellation. 
\item [Step 3.] Select the element ($S$ or $P$) with the lowest $\kappa$ from $\Lcal$, perform its 
	cancellation, recompute the pairings in $\J$, and create corresponding Jacobi regions. 
\item [Step 4.] Remove from $\Lcal$ all the existing sequences that cease to exist due to this cancellation, 
	and identify and add to $\Lcal$ any new sequences containing the newly created regions. 
\item [Step 5.] Repeat steps 3 and 4 until the Jacobi set reaches its simplest possible configuration under 
	our definition of validity or a user-defined threshold is achieved.
\end{itemize}

\paragraph{Correctness and Termination}
In order to prove the correctness of this simplification scheme, we note that 
by definition, every valid simplification step ensures that the resulting function 
$f^*$ (or $g^*$) is Morse, and the simplified Jacobi set reflects the Jacobi set of the 
simplified functions. Therefore, we have the following corollary: 

\begin{corollary}\label{lemma:validJS}
The simplified functions $f^*$ and $g^*$ are Morse and the simplified Jacobi set is a valid Jacobi set $\J(f^*,g^*)$.
\end{corollary}

\noindent 
By construction, the algorithm terminates when no other pair
of restricted critical points can be cancelled through a valid
simplification, and no other isolated critical points can be cancelled
through conventional critical point cancellation. 

In Section~\ref{app:minimal}, we study the minimal Jacobi set configuration possible for a given domain. 
Here, the minimal configuration means that the functions $f$ and $g$ are Morse functions with the minimal 
number of critical points, and the Jacobi set $\J(f, g)$ has the minimal number of loops and contains no 
BD points. 
Section~\ref{app:sphere} shows that 
for simply connected domains, our algorithm achieves this minimal configuration. Finally, Section~\ref{app:torus} 
discusses the challenges in removing certain kinds of Jacobi segments on non-simply connected domains. 
We discuss configurations with more than the minimal number of loops which cannot be simplified through 
local modification, but an extensive study on handling such cases is beyond the scope of this paper. 


\subsection{Minimal Jacobi sets}
\label{app:minimal}

As with most simplification procedures, the goal is to reach the simplest possible 
configuration according to some measure. In the case of Jacobi sets, for a pair of 
Morse functions $f$ and $g$ defined on a smooth, compact, and orientable $2$-manifold 
$\Mspace$ without boundary, the goal is to reduce the number of loops and BD points 
in $\J(f,g)$. We focus on Jacobi set with minimal configuration, i.e.\ the functions $f$ 
and $g$ are Morse functions with the minimal number of critical points, the Jacobi set 
$\J(f, g)$ has the minimal number of loops and contains no BD points. Previously, 
Bennett et al.~\cite{BennettPascucciJoy2007} suggested that given a domain with genus 
$\gamma$ the minimal Jacobi set has $\gamma + 1$ loops. We disprove this claim by 
showing that there exist functions $f$ and $g$ on $\Mspace$ that give rise to at least 
one and at most two Jacobi loops. Especially for manifold with even genus, a Jacobi set 
with one loop exists. 
Furthermore, for $\gamma = 0$ we prove that
our algorithm will reach the minimal configuration. Unfortunately, for
$\gamma > 1$ there exist configurations with more than the minimal
number of loops which cannot be simplified through local
modification. 

To proof the lemma below, 
we first give a construction of a Jacobi set containing  two loops on a (single) torus, and a single loop on a double-torus. 
Since a manifold of even genus is homeomorphic to a connected sum of double-tori, 
and a manifold of odd genus is homeomorphic to a connected sum of double-tori and a (single) torus, 
we can apply a similar construction procedure to show that there exist functions $f$ and $g$ such that $\J(f,g)$ has 
a single loop for even genus and two loops for an odd genus.

\begin{lemma}\label{lemma:minimal_loops}
  The minimal Jacobi set $\J(f,g)$ on a manifold $\M$ of genus $\gamma$ contains at 
  least $1$ and at most $2$ loops.
\end{lemma}
\begin{proof}
In the case when $\gamma = 0$, it is easy to see that there exist $f$ and $g$ that create 
only a single loop. For example, imaging a sphere embedded into $\Rspace^3$ centred at 
the origin. Then two height functions $f$ and $g$ with 90 degree angle, that is, $f(x,y,z) = x$ 
and $g(x,y,z) = z$ will create such a Jacobi set.

For $\gamma > 0$, $\M$ is homeomorphic to a connected sum of $\gamma$ tori. Such a 
surface can be constructed as the union of bent and straight cylinders as shown in 
Figure~\ref{fig:singleloop}. Imaging each piece embedded into $\Rspace^3$ with 
$g(x,y,z) = z$, the height function. Defining $f(x,y,z) = x$ creates a Jacobi set that follows
the silhouette and creates $\gamma +1$ loops. However, along a straight cylinder we can 
smoothly transition to $f(x,y,z) = -x$ (and the reverse) which winds the Jacobi set around 
the cylinder in a half turn (Figure \ref{fig:singleloop}(e)-(f)). Combining these twisted cylinders 
one can reconnect the default $\gamma +1$ loops. The Jacobi sets for the torus and 
double-torus are also shown in Figure \ref{fig:minloops} without the gluing cylinders. 

In particular, as shown in Figure~\ref{fig:singleloop}(h) for a double-torus we can connect all 
pieces into a single loop. Clearly, as shown Figure~\ref{fig:singleloop2}, by combining 
double-tori this creates functions with a single Jacobi loop for all surfaces with even genus. 
However, for a single torus the same technique simply intertwines two loops 
(Figure ~\ref{fig:singleloop}(g)). Nevertheless, treating a surface with odd genus as one with 
even genus plus a torus, there must exist $f$ and $g$ that create only two loops which proves 
the lemma.
\end{proof}  
  
\noindent We conjecture that for surfaces with uneven genus two loops is the
minimal configuration as the re-combinations must come in pairs but
currently there exists no proof.

\begin{figure*}[t]
\centering
\def\svgwidth{1.0\linewidth} 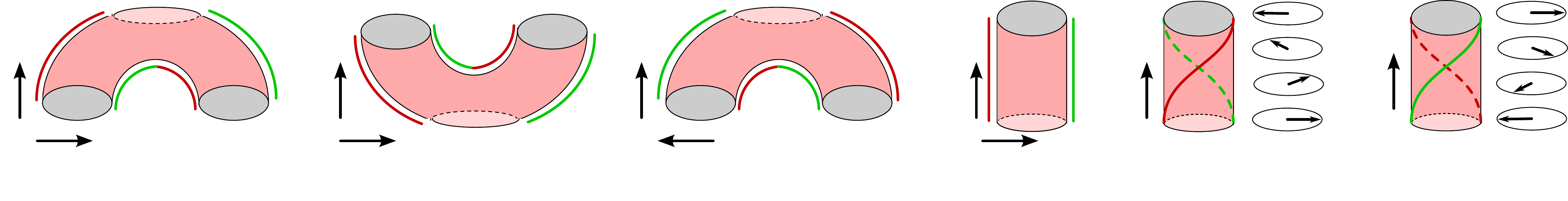

\vspace{5mm}

\def\svgwidth{0.85\linewidth} 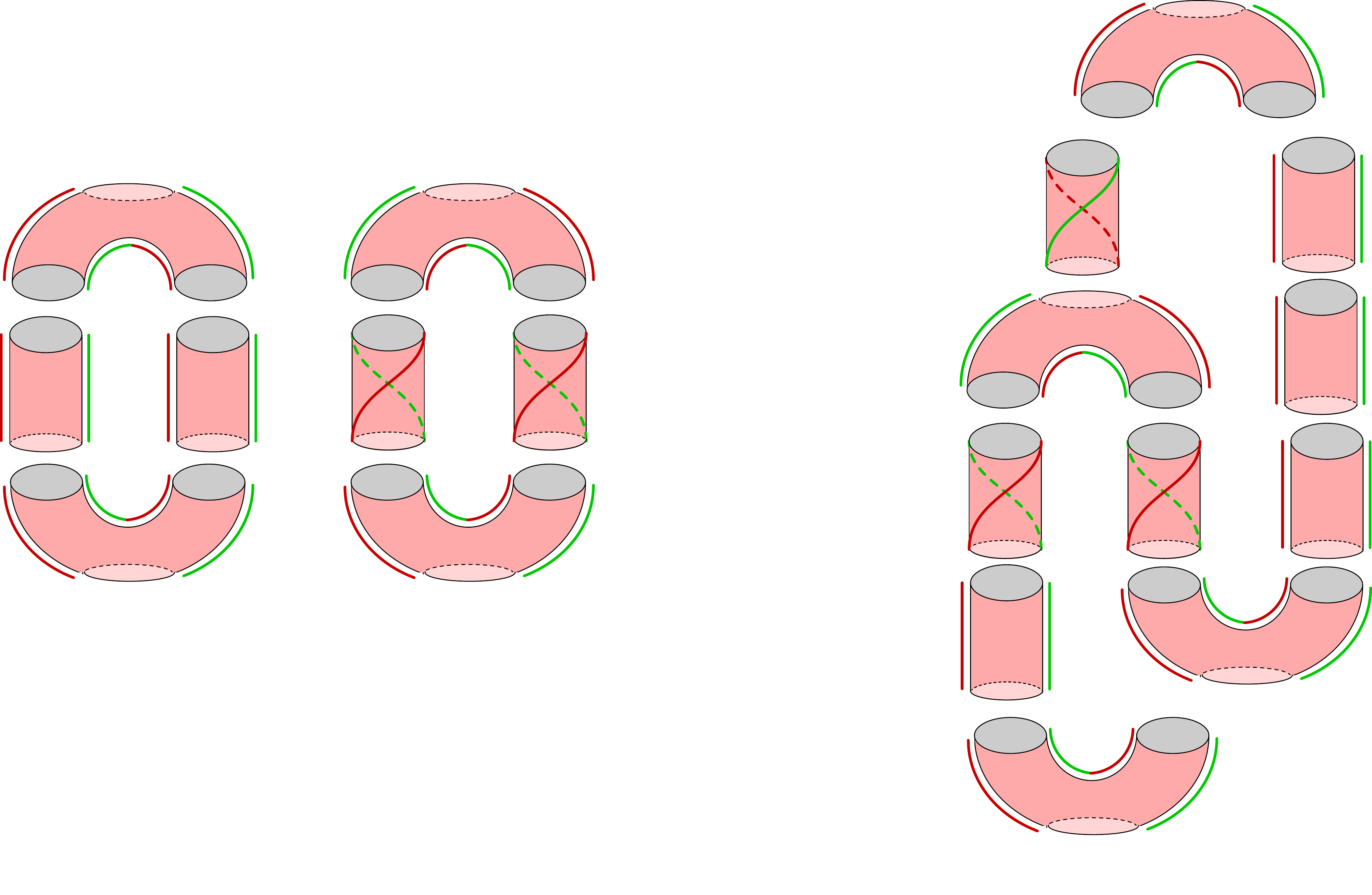
\caption{A single (g) and double-torus (h) can be constructed by gluing together six smaller pieces (a) to (f).  
The arrows indicate the directions of functions $f$ and $g$, the colors of the Jacobi loops denote criticality of 
$f_t$, and the dashed lines denote part of the loop on the back side of the manifold. 
In (e) and (f), $f$ is smoothly changed to $-f$  from top to bottom. This operation rotates the Jacobi loop between 
left and right of the corresponding pieces. When the pieces are glued together, this rotation makes it possible for 
a single Jacobi loop to connect all the critical points of $f$ and $g$ for a double-torus. However, for a single torus, 
it simply interchanges the connectivity of the two loops. 
\label{fig:singleloop}}
\end{figure*}

\begin{figure*}[!ht]
\centering
\subfigure[]{	\includegraphics[width=0.2\linewidth]{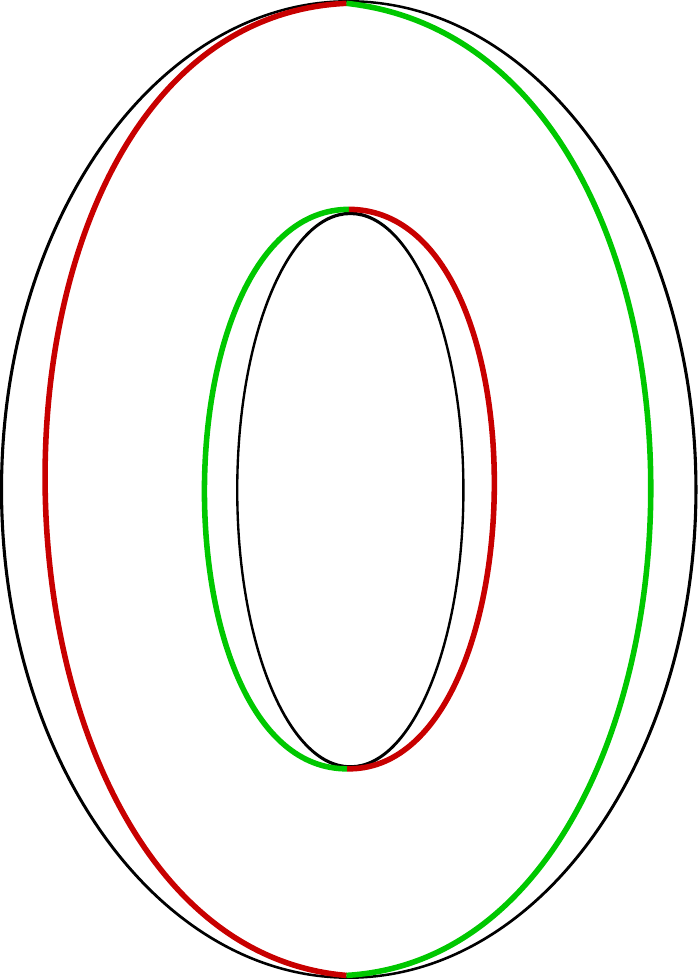} \label{fig:minloops_1}
					\includegraphics[width=0.2\linewidth]{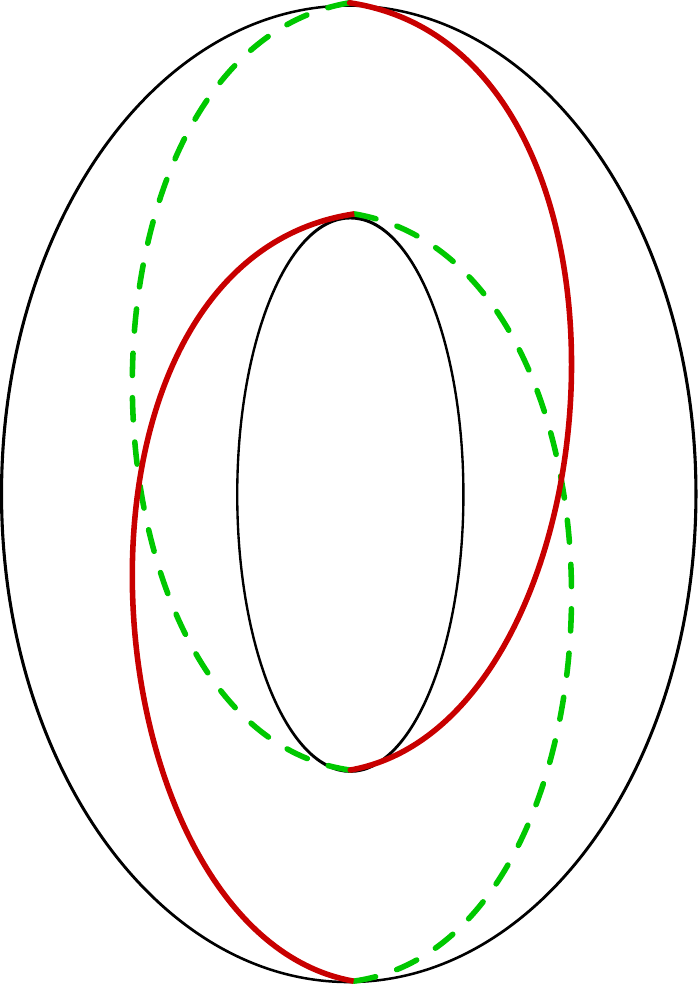} \label{fig:minloops_2}}
\hspace{2em}
\subfigure[]{	\includegraphics[width=0.2\linewidth]{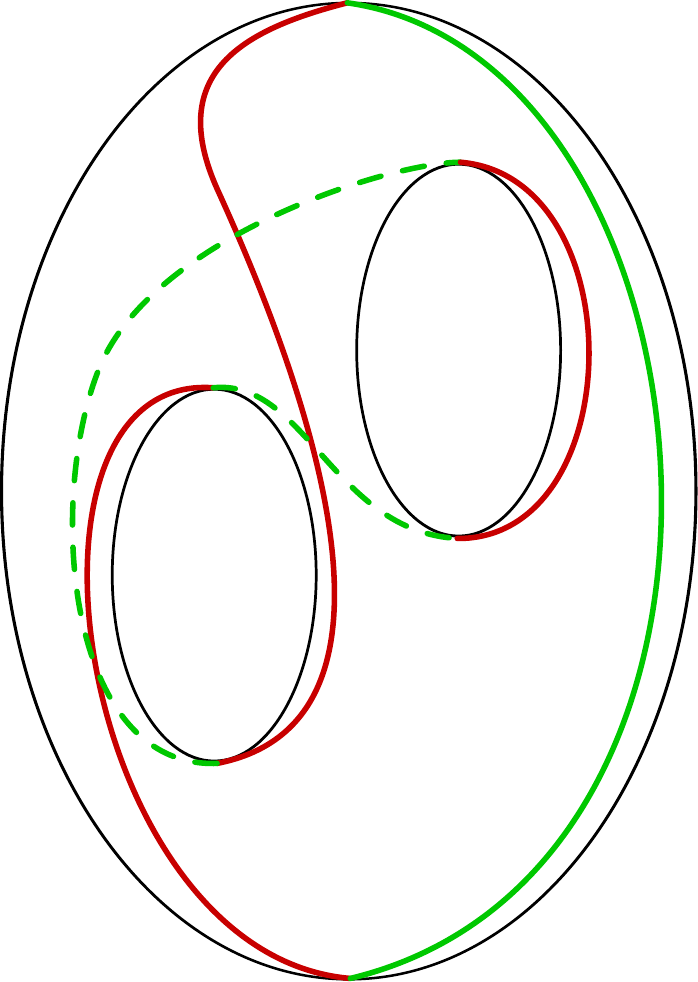} \label{fig:minloops_3}}
\caption{(a), (b) The Jacobi set on a single torus ($\gamma = 1$) contains $2$ loops, with two possible configurations. 
(c) In the case of a double torus ($\gamma = 2$), a configuration with a single Jacobi loop is feasible. 
The color of the Jacobi loops denote criticality of $f_t$, the dashed line denotes the loop on the back side of the torus.
\label{fig:minloops}}
\end{figure*}

\begin{figure*}[!ht]
\centering
\includegraphics[height=0.42\linewidth]{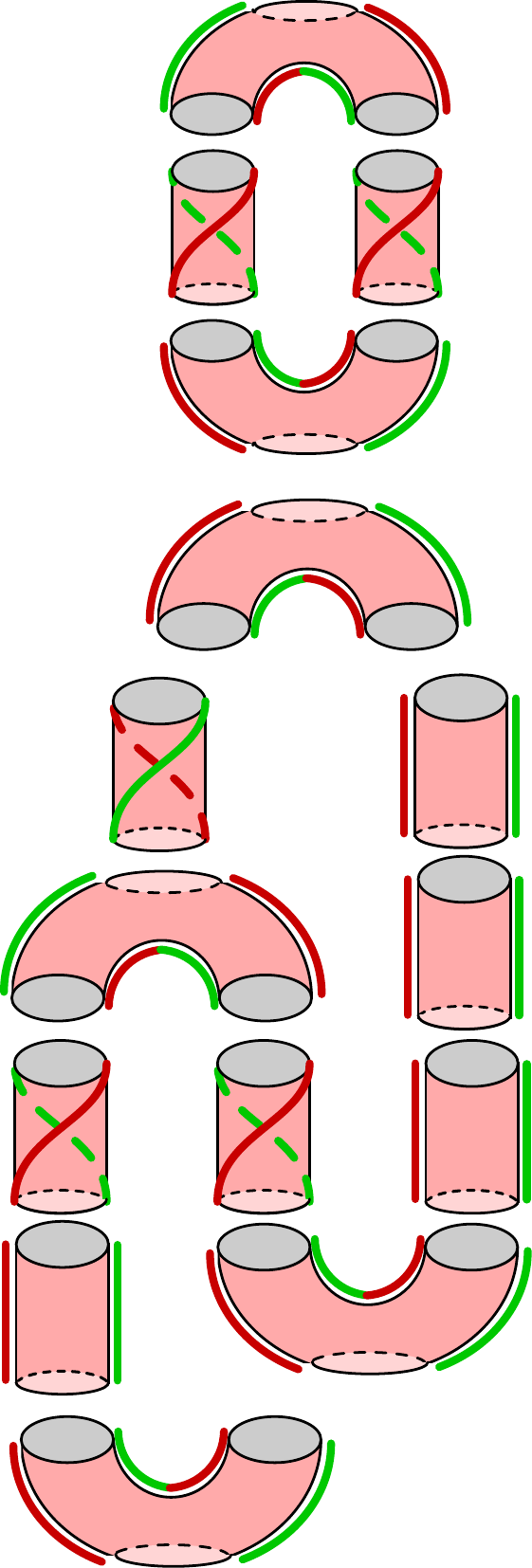}
\hspace{3cm}
\includegraphics[height=0.6\linewidth]{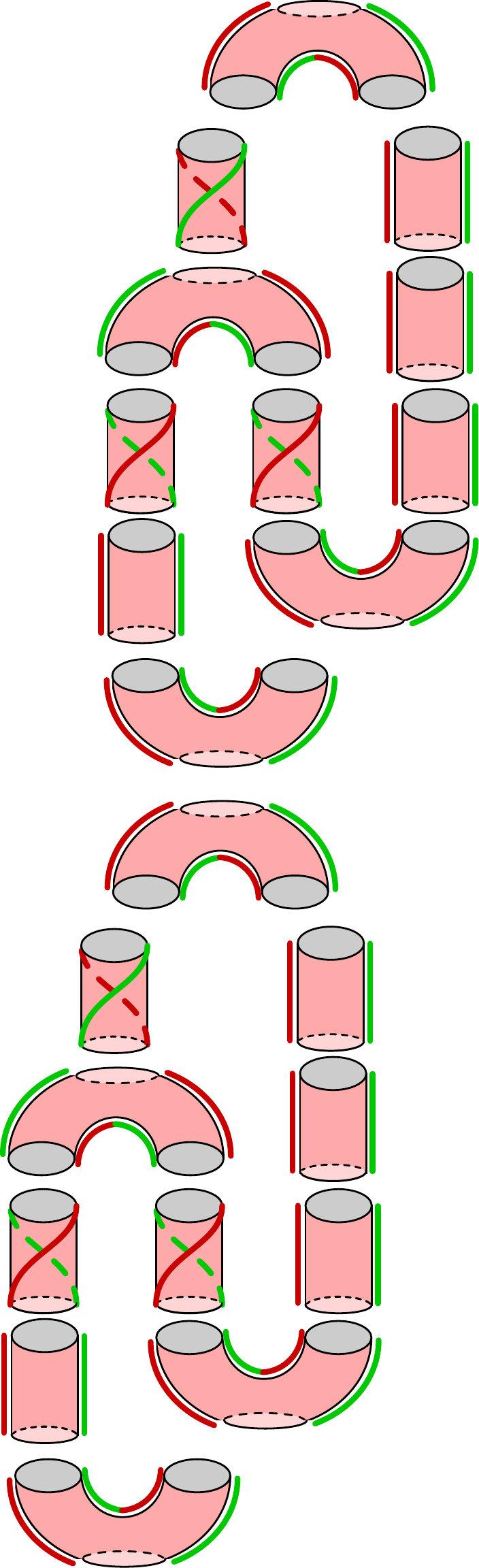}
\caption{A three-torus (left) constructed as connected sum of a single torus and a double-torus ($\Tspace^3 = \Tspace \# \Tspace^2$), 
and a four-torus (right) constructed as connected sum of two double-tori ($\Tspace^4 = \Tspace^2 \# \Tspace^2$).
Clearly, all critical points on the four-torus can be connected by the single loop. On the other hand, for a three-torus, one needs two Jacobi 
loops. 
\label{fig:singleloop2}}
\end{figure*}


\subsection{Simplification of Jacobi set on simply connected domains}
\label{app:sphere}

To show that our simplification can achieve the minimal configuration on simply connected 
domains, we first argue that if  two BD points are connected by a Jacobi loop, there always 
exists a valid sequence that removes both BD points from the Jacobi set. Furthermore, 
assuming $g$ contains only two extrema on a simply connected domain, there exist only a 
single configuration (shown in Figure~\ref{fig:proof}) such that the Jacobi set contains  
BD points not connected by the same Jacobi loop. Subsequently, we prove that these BD points 
must also be connected by a valid sequence.

\begin{figure*}[!t]
\centering
\def\svgwidth{0.4\linewidth} 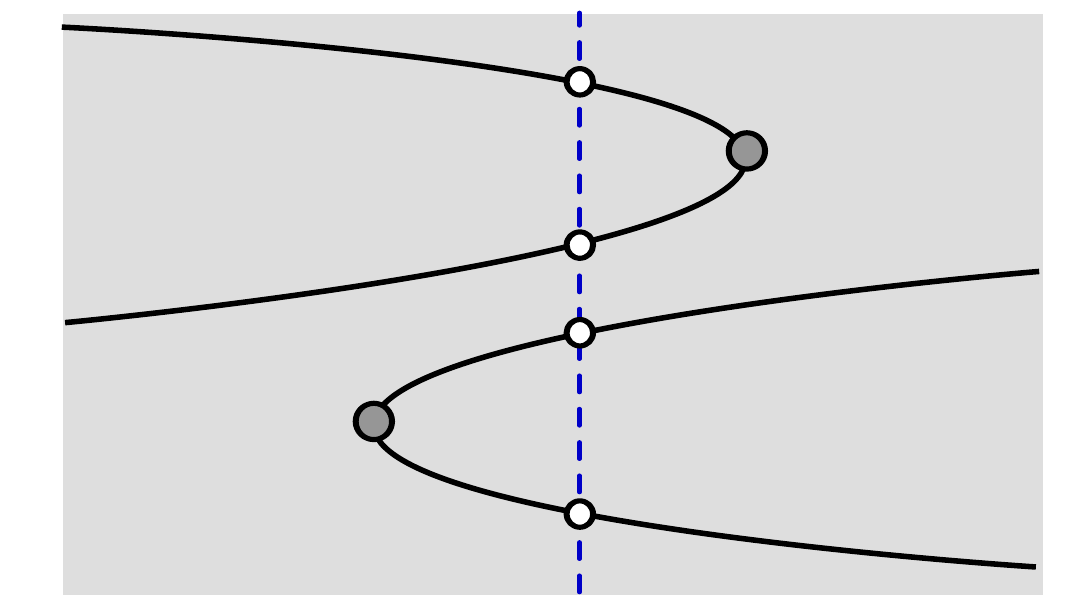 
\def\svgwidth{0.4\linewidth} 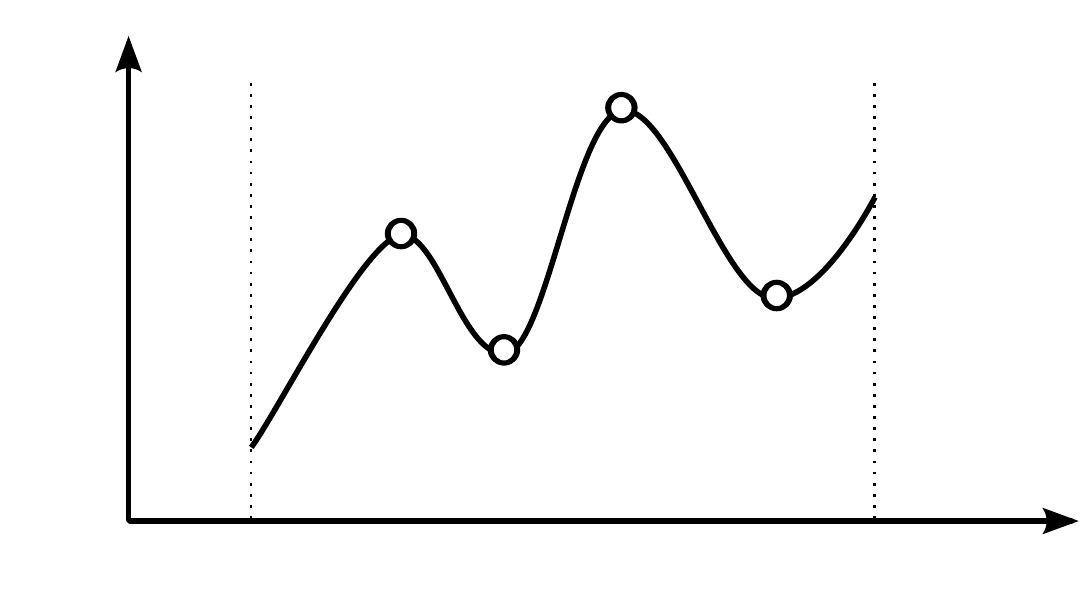
\caption{A pair of BD points must always be connected by a Jacobi sequence. Consider a level set $g^{-1}(t)$ between the level sets containing the BD points $u$ and $v$. 
Since the corresponding $f_t$ is periodic, both pairs of restricted critical points $(a_1,b_1)$, and $(a_2,b_2)$ can not be mutually paired. 
Thus, there must exist a region $R_{[t_1,t_2]}(\alpha_1,\beta_2)$, or $R_{[t_1,t_2]}(\alpha_2,\beta_1)$ leading to a valid sequence connecting $u$ and $v$.
\label{fig:proof}}
\end{figure*}

\begin{lemma}\label{lem:connected}
If $\M$ is a simply connected domain, and $u$, $v$ are two BD points 
  connected by a Jacobi loop 
  such that no critical points of $g$ or other BD points are between 
  them. Then, there exists a sequence of Jacobi regions connecting $u$ with $v$ that
  forms a valid simplification.
\end{lemma}
\begin{proof}
Let $t_1$, $t_2 \in \Rspace$ denote the function values of $g$ at the BD points, that is,  
	$t_1 = g^{-1}(u)$ and $t_2 = g^{-1}(v)$, and without loss of generality, assume $t_1 < t_2$. 
	The BD points create and destroy two restricted critical points. Since the restricted functions 
	$f_{t_0-\epsilon}$ and $f_{t_1+\epsilon}$ are Morse and hence must contain at least two 
	restricted critical points, it follows that for all 
	$t \in (t_0,t_1)$, $f_t$ has at least four restricted critical points. As a result, each point on 
	the Jacobi set connecting $u$ with $v$ is paired and can be cancelled with its
   partner. Since at a BD point, $\J$ is always mutually paired on the ``inside'' (of the BD internal region)
   there must exist a valid sequence or Jacobi regions connecting $u$ with $v$. 
   
   Note that the two possible configurations for this scenario are shown in Figures~\ref{fig:valid} and~\ref{fig:loop_cancel2}.
\end{proof}

\noindent
To prove the main result we note that on a simply connected domain all saddles
 of $f$ and $g$ can be removed either through simplifying the Jacobi set or
 through  direct cancellations. 
 That is, on a simply connected domain all saddles of $f$
and $g$ can be removed through cancellations and only a single
minimum and a single maximum remains. In this case, any potentially
remaining BD points must form a valid sequence of regions, as no
critical points exist that may block a sequence from being formed.

\begin{lemma}
\label{lemma:algo_simplest}
If $\Mspace$ is a simply connected domain ($\gamma = 0$), the algorithm reduces a Jacobi set to 
its minimal configuration -- a single loop without birth death points.
\end{lemma}
\begin{proof}
  Without loss of generality, we suppose $f$ and $g$ contain no saddles, 
  since for simply connected domain, all saddles can be cancelled (either through 
  sequence cancellation or by direct cancellation with extrema of the functions). 
  As a result the level sets of $g$ can be seen as a
  collection of vertical lines periodic at $\infty$ as shown in
  Figure~\ref{fig:proof}. Following Lemma~\ref{lem:connected}, all BD points
  connected by Jacobi loops can be removed through valid cancellations. Nevertheless,
  there can exist two BD points are not connected to each other, one forming a loop
  with the  maximum of $g$ and one forming a loop with the minimum of
  $g$. Both loops must overlap since each $f_t$ must have at least two
  restricted critical points.

  Assume there does not exist a sequence connecting the two BD points. It
  follows that the curve $\alpha_1$ is always mutually paired with $\beta_1$ and
  $\alpha_2$ is always mutually paired with $\beta_2$. We show that this is a
  contradiction. Assume $a_1$, $a_2$ are maxima and $b_1$, $b_2$ are minima. If
  $\alpha_1$ is mutually paired with $\beta_1$ then $f(b_1) > f(b_2)$ and
  $f(a_1) < f(a_2)$. (Remember that the level sets are periodic.) However,
  $\alpha_2$ mutually paired with $\beta_2$ implies $f(a_2) < f(a_1)$ which gives a
  contradiction, and hence proves the lemma.
\end{proof}


\subsection{Simplification of Jacobi set on non-simply connected domains\label{app:torus}}

On simply connected domains, we showed that our simplification scheme can obtain the minimal Jacobi set 
configuration. Here, we discuss the fundamental problems that are inherent to the structure of 
non-simply connected domains, and how they impact our simplification algorithm. 
The simplification algorithm terminates when no more valid simplification steps are possible, and hence 
no more Jacobi sequences can be found. 
For non-simply connected domains, there exist saddles that can not be removed. These saddles may 
block the construction of Jacobi sequences, and hence no more Jacobi sequences may be formed. Thus, our 
algorithm may terminate without achieving the minimal Jacobi set.

To contrive such an example, we start with the minimal Jacobi set (shown in Figure~\ref{fig:minloops_1}) on 
a torus $\Tspace$, where the two functions are height functions with an angle $\pi/2$ between them. 
The function $f$ can then be changed along the outer silhouette of the torus, using a sinusoidal kernel that 
replaces the restricted maxima with a valley and restricted minima with a ridge. For each $f_t$, this operation 
replaces one restricted critical point by three, thus creating two extra Jacobi loops. Since the function must 
stay smooth, the kernel must go to zero at the critical points of $g$, where the restricted critical points of 
$f_t$ switch. 

To understand this Jacobi set, recall that the torus is constructed as the product of two circles. If $\theta$ and $\phi$ denote the 
polar angle of the two circles, then the torus can be parametrized as $\Tspace(\theta,\phi)$. Figure~\ref{fig:torus_imp} 
shows the 
level sets and critical points of the two functions (in red and blue) on the $\theta-\phi$ plane along with 
the Jacobi set (in black). Clearly, there exist four loops in the Jacobi set. The saddles on $\J_1$ and $\J_3$ also 
act as BD points. Any sequences that are seeded at the BD points always get stuck at the saddles and hence, 
no valid sequence is possible. 

Notice that in this case, there exists no simplification that can be achieved by local modification to the 
functions. Thus, our algorithm can not simplify this Jacobi set further. Going forward, we envision more 
general and global simplifications steps, which modify more than two loops of Jacobi set simultaneously. 
Such simplifications will be able to handle such difficult cases for non-simply connected domains. 

\begin{figure*}[!ht]
\centering
	\def\svgwidth{0.6\linewidth} 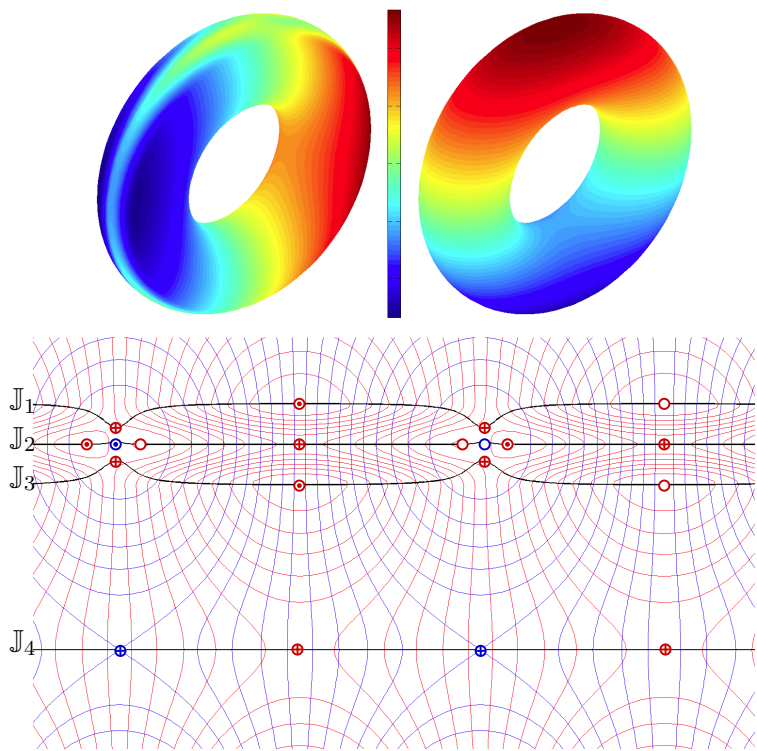 
 \caption{(Top) Functions $f$ (left) and $g$ (right) are defined on a torus, $\Tspace(\theta,\phi)$. (Bottom) The level sets and critical points of 
 $f$ and $g$ are shown in red and blue respectively, along with the Jacobi set in black, on the $\theta-\phi$ plane. 
 Since the domain is periodic, the four Jacobi loops are closed. Although, there exist BD points 
 on $\J_3$ and $\J_4$ (coinciding with saddles of $f$), the algorithm can not find a valid Jacobi sequence due to the 
 presence of irremovable saddles.\label{fig:torus_imp}}
\end{figure*}

\section{Discussion and Future Work\label{sec:discussion}}

In this paper, we introduce a direct technique for Jacobi set simplification, 
aimed at achieving local, smooth, and consistent modifications to the 
underlying functions. Our technique is guided by a user-defined 
metric, and offers fine control over the
simplification process and could be widely applicable in many data analysis
applications.  While our choice of the metric -- the comparison measure $\kappa$
is well-suited for our purpose, we would like to explore other possibilities.
Further, while it is clear that the algorithm reduces a Jacobi set to its
simplest configuration for simply connected domains, there exist cases where
this is not possible for non-simply connected domains. Understanding of such cases 
may lead to the need of global simplification operations which can help obtain 
the simplest Jacobi set for manifolds with non-zero genus. 
We wish to explore such cases and extend our simplification scheme 
to address them. 
Lastly, while the focus of the current work is a detailed discussion on the various 
elements of the algorithm for smooth functions, a technical discussion on its 
discrete adaptation and practical implementation is forthcoming. 


\section*{Acknowledgement}
\noindent The authors would like to thank Attila Gyulassy for insightful discussions during the early stage of this project. 

\bibliographystyle{abbrv}
\bibliography{refs}


\end{document}